\PassOptionsToPackage{shortlabels}{enumitem}
\documentclass[10pt,twocolumn,letterpaper]{article}
\usepackage{graphicx} 
\usepackage[dvipsnames,table]{xcolor}  

\usepackage[ruled,vlined]{algorithm2e}
\usepackage{amsmath}
\usepackage{float}
\usepackage{amssymb}
\usepackage{amsthm}

\usepackage[utf8]{inputenc} 
\usepackage[T1]{fontenc}    
\usepackage{hyperref}       
\usepackage{url}            
\usepackage{booktabs}       
\usepackage{amsfonts}       
\usepackage{nicefrac}       
\usepackage{microtype}      
\usepackage{bbm}
\usepackage{graphicx}
\usepackage{caption}
\usepackage{subcaption}
\usepackage{wrapfig}
\usepackage{multirow}
\usepackage{comment}
\usepackage{enumitem}
\usepackage{colortbl}
\usepackage{pifont}

\newtheorem{definition}{Definition}
\newtheorem{lemma}{Lemma}
\newtheorem{theorem}{Theorem}

\usepackage{subcaption} 
\usepackage{enumitem}

\usepackage[pagenumbers]{iccv} 
\title{DuoLoRA : Cycle-consistent and Rank-disentangled Content-Style Personalization}

\author{Aniket Roy$^{1,2}$\thanks{Work done as part of a summer internship.}, Shubhankar Borse$^{2}$, Shreya Kadambi$^{2}$, Debasmit Das$^{2}$, Shweta Mahajan$^{2}$, \\
Risheek Garrepalli$^{2}$, Hyojin Park$^{2}$, Ankita Nayak$^{2}$, Rama Chellappa$^{1}$, Munawar Hayat$^{2}$, Fatih Porikli$^{2}$ \\
{$^{1}$Johns Hopkins University, $^{2}$Qualcomm AI Research}\thanks{Qualcomm AI Research is an initiative of Qualcomm Technologies,
Inc.}}

\begin{document}

\maketitle
\begin{abstract}
We tackle the challenge of jointly personalizing content and style from a few examples. A promising approach is to train separate Low-Rank Adapters (LoRA) and merge them effectively, preserving both content and style. Existing methods, such as ZipLoRA, treat content and style as independent entities, merging them by learning masks in LoRA's output dimensions. However, content and style are intertwined, not independent. To address this, we propose DuoLoRA—a content-style personalization framework featuring three key components: (1) rank-dimension mask learning, (2) effective merging via layer priors, and (3) Constyle loss, which leverages cycle-consistency in the merging process.
First, we introduce ZipRank, which performs content-style merging within the rank dimension, offering adaptive rank flexibility and significantly reducing the number of learnable parameters. Additionally, we incorporate SDXL layer priors to apply implicit rank constraints informed by each layer’s content-style bias and adaptive merger initialization, enhancing the integration of content and style. To further refine the merging process, we introduce Constyle loss, which leverages the cycle-consistency between content and style.
Our experimental results demonstrate that DuoLoRA outperforms state-of-the-art content-style merging methods across multiple benchmarks.
\end{abstract} 
\vspace{-0.4cm}
\section{Introduction}
\label{sec:intro}
\begin{figure}
    \centering
        \centering
        \scalebox{0.43}{
        \includegraphics[width=1.0\textwidth]{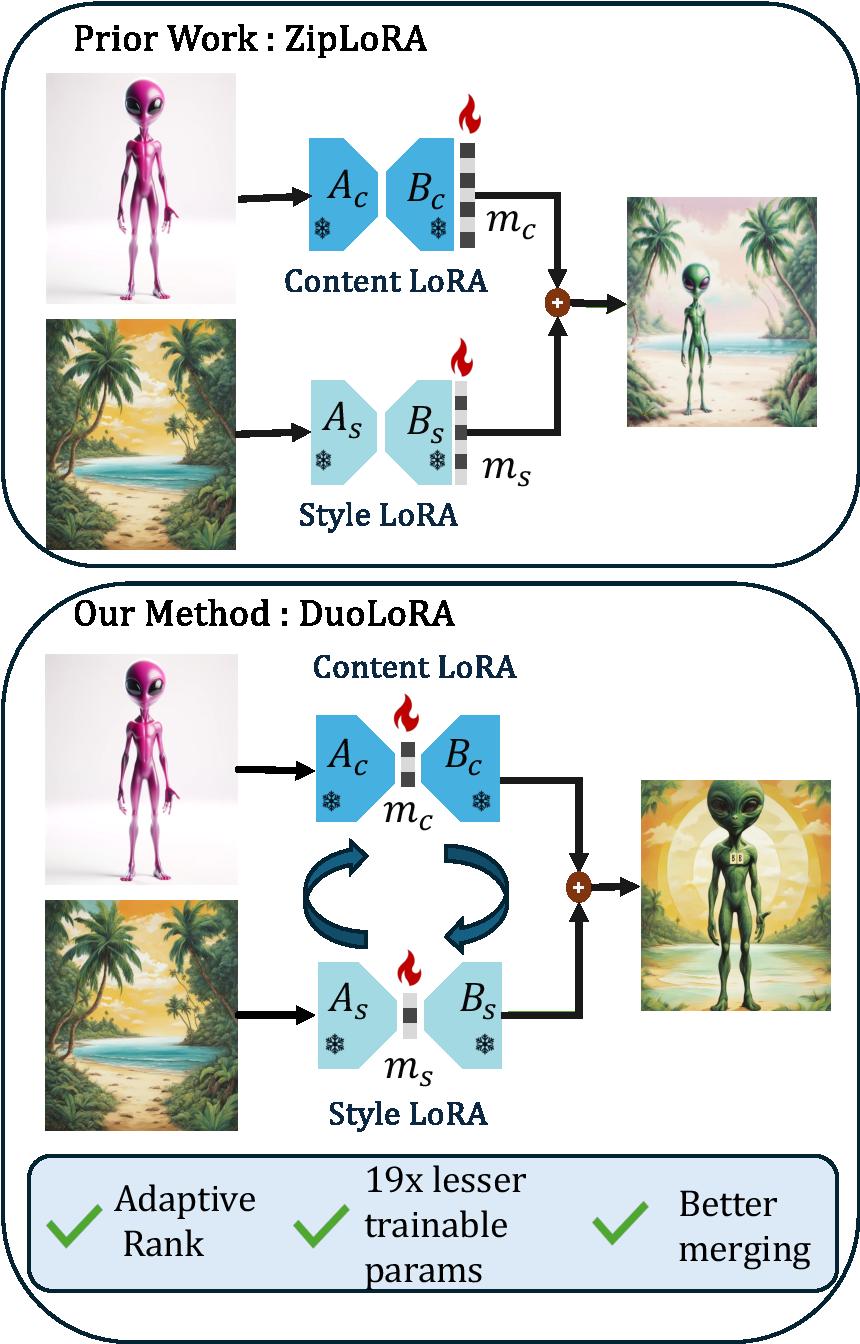}}
        \vspace{-0.2cm}
        \caption{Content and style personalization using DuoLoRA provides (1) adaptive rank flexibility, (2) significantly lesser trainable parameters, and (3) better content-style merging.}
        \label{fig:teaser}  
        \vspace{-0.6cm}
 \end{figure}

Personalizing text-to-image diffusion models~\cite{ruiz2023dreambooth,sohn2023styledrop,chang2023muse} has recently attracted significant interest due to their broad range of applications. 
These models can produce images with desired content and style, overcoming the limitations of earlier approaches that required extensive training data for neural style transfer~\cite{gatys2016image}. 
They can generate specific images from just one or two reference images, but blending both content and style remains challenging.

Recently, parameter-efficient fine-tuning (PEFT) methods like Low-Rank Adapters (LoRA)~\cite{hu2021lora} have become popular. 
These methods can capture unique characteristics with a small amount of data. 
Consequently, content-style personalization can be reframed as a LoRA merging task, where individually trained LoRAs are combined to achieve efficient content-style blending. 
For example, ZipLoRA~\cite{shah2025ziplora} merges content and style by learning masks over the LoRA output dimensions, assuming that content and style are independent concepts. 
This approach provides adequate control over the diffusion model. 
However, the requirement for fine-tuning during inference is a drawback. Additionally, using the same rank for both content and style LoRAs may not be optimal, as content and style may have different representational requirements depending on each layer.


In this work, we ask: How can we enable adaptive rank flexibility to reduce fine-tuning costs while increasing the separation between content and style distributions?
To achieve adaptive rank flexibility across the layers of the diffusion UNet, we analyze the SDXL model architecture~\cite{basu2024mechanistic, frenkel2024implicit}. We observe that UNet layers, with smaller resolutions, primarily influence content generation, suggesting that content merges should have higher rank in these layers than style merges. Conversely, the layers, with larger resolutions, are more crucial for style generation, and the rank constraint should be adjusted accordingly~\cite{basu2024mechanistic, frenkel2024implicit}.

Within this context, we introduce DuoLoRA, a framework for effective content-style merging, composed of three components: (1) ZipRank, which learns masks within the rank dimension, (2) improved merging via layer priors, and (3) Constyle loss, which leverages cycle-consistency across content and style.
In ZipRank, we propose learning masks in the rank dimension rather than the output dimension, allowing for adaptive rank adjustment and a substantial reduction in learnable parameters. To facilitate layer-wise rank-adaptive merging, we apply rank constraints during content-style merging, informed by our observations as prior knowledge. This rank constraint is formulated as a nuclear norm minimization problem under a $l1$ sparsity constraint, improving content-style blending with fewer parameters.
To further disentangle content and style, we introduce Constyle loss, which leverages cycle-consistency across content and style. This approach, inspired by CycleGAN’s~\cite{zhu2017unpaired} treatment of content and style as separate domains optimized through domain translation with minimized reconstruction loss, enables balanced content-style merging. Unlike ZipLoRA, which treats content and style independently, our cycle-consistency loss accounts for their interdependent nature, resulting in improved content-style blending.

In summary, the contributions of this paper are as follows:

\begin{itemize} 
    \item We address content-style personalization as a LoRA merging problem, where we propose learning masks in the rank dimension instead of the output dimension, allowing for adaptive rank flexibility with significantly fewer learnable parameters. 
    \item We analyze the SDXL architecture’s layer prior information, finding that layers with lower resolutions primarily contribute to content generation, while layers with higher resolutions focus on style generation. Based on these insights, we introduce explicit rank constraints through nuclear norm minimization under a sparsity constraint to improve merging. 
    \item We introduce Constyle loss, leveraging the cycle-consistency across content and style, and we validate this approach across various benchmarks. 
\end{itemize} 
\section{Related work}
\label{sec:related_work}

\textbf{Personalization.}
Personalizing text-to-image diffusion models has recently attracted significant attention. Early approaches, such as Textual Inversion~\cite{gal2022image}, focus on learning a text token that represents a particular concept, while Dreambooth~\cite{ruiz2023dreambooth} updates network parameters for personalization. Custom Diffusion~\cite{kumari2023multi} makes the process more efficient by fine-tuning only the cross-attention modules. However, these methods typically handle only a single concept or object. In contrast, some approaches aim to personalize both content and style. StyleDrop~\cite{sohn2023styledrop} uses the Muse transformer network to align the style of generated images with a reference image. Other recent style learning methods are - Rb-modulation~\cite{rout2024rb}, Instantstyle~\cite{wang2024instantstyle}, IP-adapter~\cite{ye2023ip}, Magic-insert~\cite{ruiz2024magic}, StyleAlign~\cite{hertz2024style}, LoRA-composer~\cite{yang2024lora}, Paircustomization~\cite{jones2024customizing},~\cite{shenaj2024lora,xu2024freetuner,xu2024break,zhang2024finestyle,ouyang2025k,hu2021lora,choi2024style,bi2024customttt,avrahami2024diffuhaul,xie2024styletex,cohen2024conditional,stoica2023zipit,wang2024instantstyle,xing2024csgo,jiang2024artist,ohm2025fruit,kompanowski2024dream,agarwal2024training,ge2024tuning,hu2024vectorpainter,song2024style3d,deng2024magicstyle,sunsast,borse2025subzero,somepalli2024measuring,liu2024unziplora}. 

\noindent\textbf{LoRA merging.}
LoRA~\cite{hu2021lora} has proven particularly effective for learning from small datasets. Individual LoRA models are trained for each concept or style, and these concepts can then be combined through LoRA merging. Methods like Concept-Sliders~\cite{gandikota2023concept} and ControlNet~\cite{zhang2023adding} also utilize LoRA merging, though their approaches are primarily suited for text-based editing. Recently, Shah et al.~\cite{shah2025ziplora} introduced a method for LoRA merging by learning orthogonal masks in the output space of the layer. 

\noindent\textbf{Layer priors.}
We analyze the SDXL architecture and uncover some intriguing insights. Specifically, we observe that the earlier layers in the diffusion UNet contribute more to local edits, such as content generation, while the later layers influence more global edits. Other baseline approaches include naive parameter merging~\cite{shah2025ziplora} and B-LoRA~\cite{frenkel2024implicit}. Recent studies have also examined how information is encoded within the SDXL architecture, demonstrating its utility for fast image editing~\cite{basu2024mechanistic}. These findings inspire us to investigate layer prior information-based merging in this work.
\section{Method : DuoLoRA}
\label{sec:method}
\subsection{Problem statement}

\begin{figure*}
    \centering
        \centering
        \includegraphics[width=0.98\textwidth]{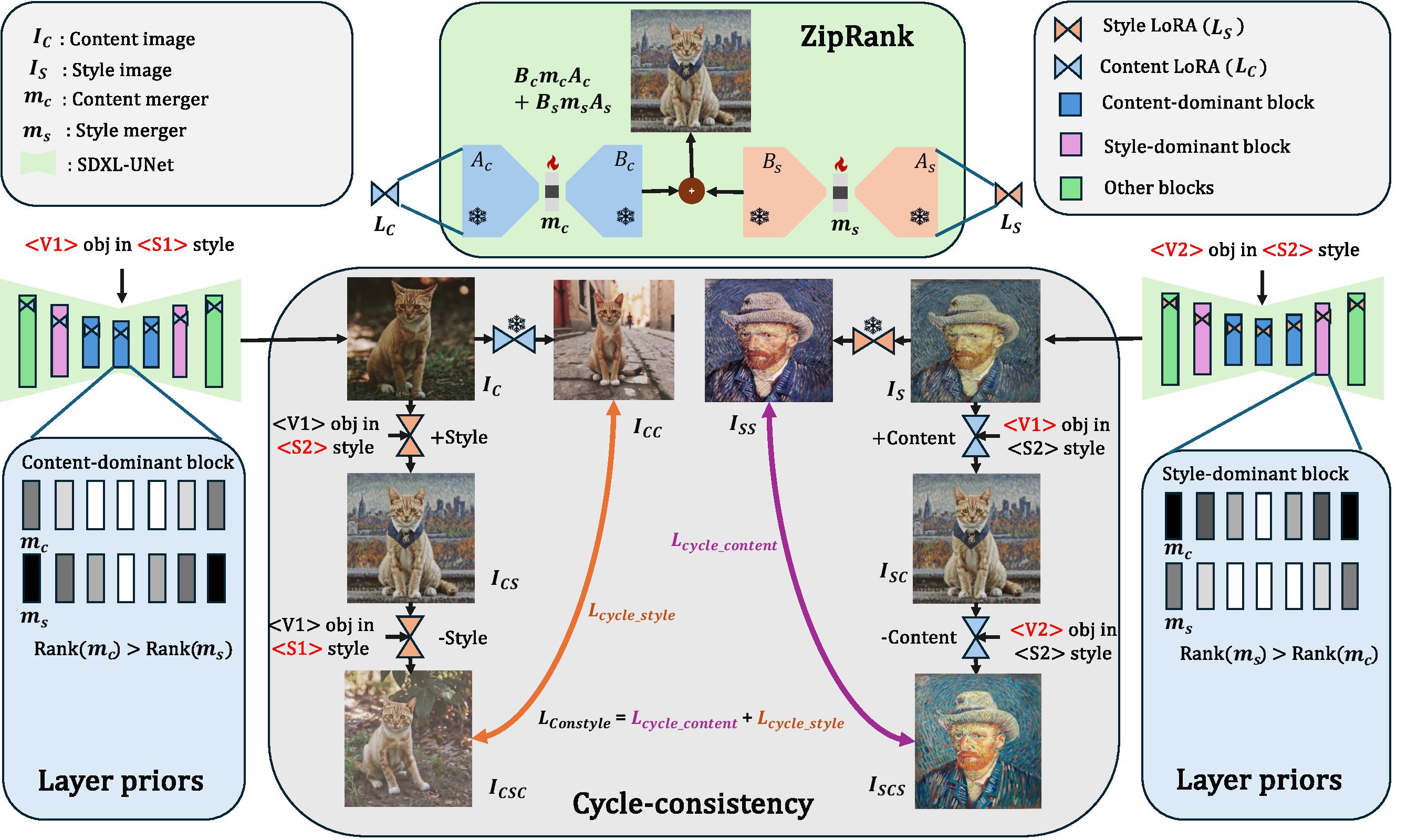} 
        \caption{Overview of DuoLoRA. It consists of three components - (1) ZipRank: learning the mask in rank dimension, (2) layer-prior based merging identifying content-dominant and style-dominant blocks of SDXL UNet, (3) cycle-consistency based merging using Constyle loss.}
        \vspace{-0.4cm}
        \label{fig:method}   
 \end{figure*}

We address the problem of customizing both content and style from few examples and generate variations conditioned on text prompt. We investigate the problem through the lens of merging individual concepts, where such individual concepts are learned using LoRAs. Therefore we cast the concept-style personalization as a LoRA merging problem (as shown in Fig.~\ref{fig:teaser}).
However, content and style are not orthogonal, rather those are intertwined concepts. Therefore, entangling such concepts remains challenging. 

In order to efficiently merge content and style specific LoRAs, we first learn masks in the rank dimension, with UNet layer-prior informed initialization, explicit rank and sparsity constraint. To further improve performance, we use content-style cycle-consistent merging. We explain this as follows.

\subsection{ZipRank: Merging in Rank space}

We start with providing background on LoRA.
\newline
\textbf{Low-Rank Adaptation (LoRA): }
In LoRA~\cite{hu2021lora}, we fine-tune neural networks by approximating weight updates using low-rank matrices. Specifically, the weight update matrix $\Delta W$ is parameterized as: $\Delta W = A B = U_r \Sigma_r V_r^\top$, 
where $A \in \mathbb{R}^{d_{\text{out}} \times r}, B \in \mathbb{R}^{r \times d_{\text{in}}},$ and $r$ is the rank of the approximation.

Blending LoRAs optimally presents a significant challenge. A straightforward approach is to use simple arithmetic merging, but this is not efficient.
ZipLoRA~\cite{shah2025ziplora} (Fig.~\ref{fig:teaser}) uses the Zipit~\cite{stoica2023zipit} operation across model weights through learnable masks in the output dimension as defined below, ensuring the masks in the weight space remain orthogonal.

\begin{definition}{Output Dimension Masking (ZipLoRA~\cite{shah2025ziplora}): }
Here we apply a mask to the output dimension by defining a diagonal mask matrix $M_{\text{out}} \in \mathbb{R}^{d_{\text{out}} \times d_{\text{out}}}$ with entries $n_{ii} \in \{0, 1\}$. The output-masked approximation is:
\[
\Delta W_{\text{out}} = M_{\text{out}} A B = M_{\text{out}} U_r \Sigma_r V_r^\top.
\]
Let $d_s = \sum_{i=1}^{d_{\text{out}}} n_{ii}$ be the number of active output units.
\end{definition}

Our approach, however, focuses on learning masks in the rank dimension (ZipRank, Fig.~\ref{fig:method}), which greatly reduces the number of learnable parameters and provides adaptive rank flexibility across the LoRA adapters.

\begin{definition}{Rank Dimension Masking (ZipRank): }
We apply a mask to the rank dimension by defining a diagonal mask matrix $M_r \in \mathbb{R}^{r \times r}$ with entries $m_{ii} \in \{0, 1\}$. The rank-masked approximation is:
\[
\Delta W_{\text{rank}} = A M_r B = U_r M_r \Sigma_r V_r^\top.
\]
Let $S$ denote the set of indices $i$ where $m_{ii} = 1$, and $s = |S|$ is the number of active rank components.
\end{definition}

ZipRank is designed to separate the style and content adapters while minimizing the number of learnable parameters. This approach reduces overfitting, especially when training on small datasets, and introduces adaptive rank flexibility during the merging of LoRAs. We have also shown that under the same parameter budget, the approximation error resulting from rank dimension masking is less than or equal to that from output dimension masking in Theorem~\ref{theroem:output_rank_merging}.

\begin{theorem} In LoRA merging, under the same parameter budget, the approximation error resulting from rank dimension masking is less than or equal to that from output dimension masking. Formally, $E_{\text{rank}} \leq E_{\text{out}},$
where the approximation error using rank dimension masking ($E_{\text{rank}}$), and output dimension masking ($E_{\text{out}}$) are given by,
\[E_{\text{rank}} = \| X - \Delta W_{\text{rank}} \|_F, \quad E_{\text{out}} = \| X - \Delta W_{\text{out}} \|_F\]
\label{theroem:output_rank_merging}
\end{theorem}
\vspace{-1cm}
\begin{proof}
Proof is provided in the supplementary material.
\end{proof}

\subsection{Layer priors}

\begin{figure}
    \centering
        \centering
        \includegraphics[width=0.45\textwidth]{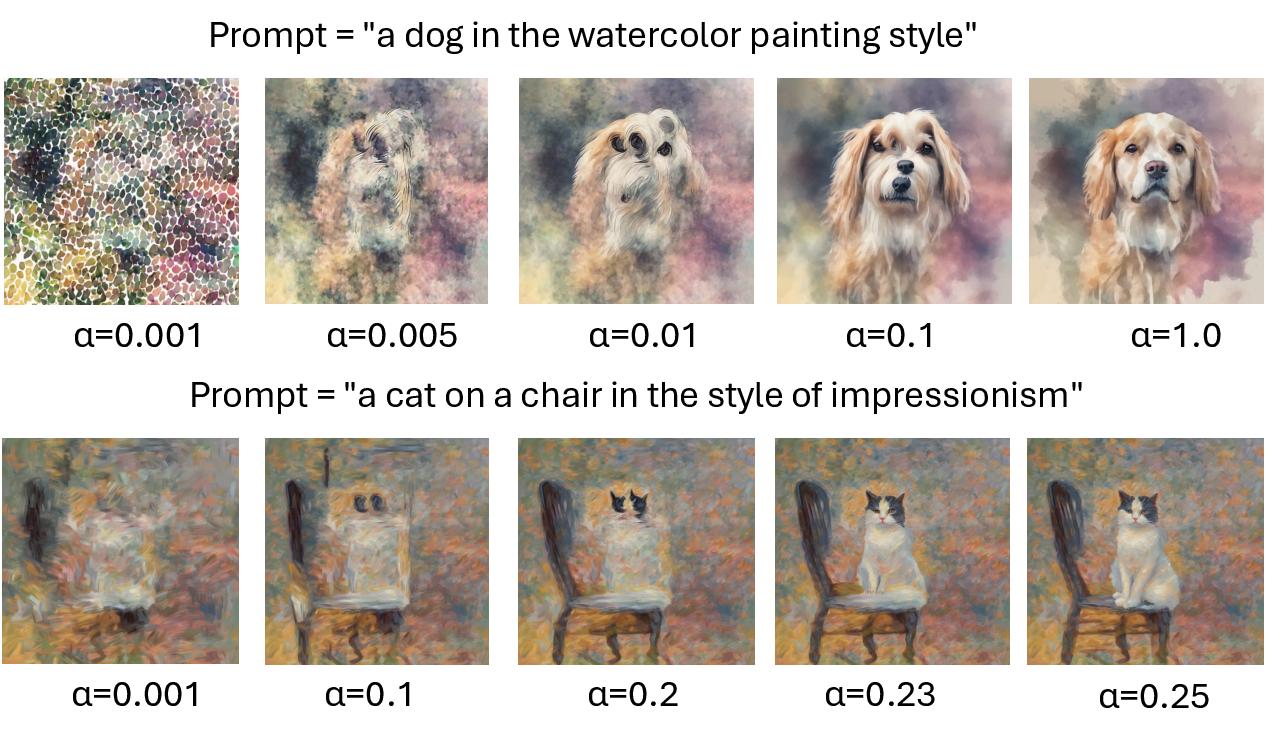} 
        \vspace{-0.2cm}
        \caption{SDXL selective weight scaling. The scaling parameter ($\alpha$)
has been applied to the content-dominant blocks (\texttt{up\_block.2},
\texttt{down\_block.2}, and \texttt{mid\_block}). We observe that with a
smaller $\alpha$, the model fails to generate the content and instead
focuses solely on the style. In contrast, increasing $\alpha$ allows the
model to generate the content specified in the prompt.}
        \vspace{-0.4cm}
        \label{fig:sdxl_alpha}   
 \end{figure}

\subsubsection{How content style is encoded in SDXL}
\label{subsec:understanding_sdxl}

Understanding how content, style, and similar information are encoded within the diffusion UNet is crucial for enabling both local and global edits. 
We hypothesize that in SDXL, layers with low resolutions (e.g.,\texttt{up\_block.2, down\_block.2, mid\_block}, with resolution < 32) capture localized updates, such as content details, while layers with high resolutions (e.g.,\texttt{up\_block.1, down\_block.1}, with resolution >= 32) capture global updates, such as the overall style of the image.

This hypothesis is motivated by the structural design of the SDXL model: layers with low resolutions, typically in the earlier stages of the U-Net, focus on small, localized regions, allowing them to encode essential spatial details, shapes, and structures. This fine-grained attention enables accurate representation of core content elements, like objects and scene layouts, aligning closely with the input prompt.

In contrast, layers with high resolutions, often in deeper sections of the model, capture style by blending and harmonizing details across larger regions. This broader scope makes them well-suited for encoding global stylistic features like color gradients, textures, and lighting. Such attributes require a coherent application across the image rather than precision in spatial detail. Consequently, high resolutions are believed to play a crucial role in generating the overall aesthetic and mood of the image, supporting their role in style encoding.

To verify the hypothesis, we use a layer-freezing simulation via weight scaling. This method allows us to selectively reduce the contributions of specific layers during inference, enabling us to observe the distinct roles that different resolutions play in image generation.

We begin by identifying the layers in the SDXL U-Net architecture associated with low and high resolutions. Layers such as \texttt{up\_block.2}, \texttt{down\_block.2}, and \texttt{mid\_block} have resolutions smaller than 32, while layers like \texttt{up\_block.1} and \texttt{down\_block.1} have resolutions larger than 32.

Next, we apply selective weight scaling to simulate a ``freezing'' effect. We scale the outputs of either low- or high-resolution layers by a small factor, $\alpha$ (e.g., 0.1), to reduce their influence in the generation process. Scaling down the low-resolution layers diminishes their impact on content details, while scaling down the high-resolution layers lessens the effect on style features.

For instance, when using the prompt, \texttt{p = ``A cat on a chair in the style of impressionism''} during inference, we scale down the weights of the low-resolution layers (\texttt{up\_block.2}, \texttt{down\_block.2}, and \texttt{mid\_block}) by a scalar factor $\alpha$. By varying $\alpha$, as shown in Fig.~\ref{fig:sdxl_alpha}, we observe that lower values of $\alpha$ generate the style but fail to capture the object (i.e., the ``cat''). As we increase $\alpha$, the object becomes visible in the generated image. This trend, illustrated in Fig.~\ref{fig:sdxl_alpha}, supports our hypothesis that low-resolution layers (\texttt{up\_block.2}, \texttt{down\_block.2}, and \texttt{mid\_block}) contribute to content generation, while layers with higher resolution (\texttt{up\_block.1} and \texttt{down\_block.1}) contribute to style. These observations align with prior findings in the literature~\cite{basu2024mechanistic}.




Recent works, like B-LoRA~\cite{frenkel2024implicit}, have identified specific layers involved in content (W4 in SDXL) and style (W5 in SDXL). However, these findings are highly specific and may not generalize widely.
Rather than adhering strictly to these findings, we incorporate this knowledge in a more generalized way by blending content-style LoRAs.
In SDXL, we observe that layers \texttt{up\_block.2}, \texttt{down\_block.2}, and \texttt{mid\_block}, which have lower resolutions ($\leq$ 32), are more involved in content generation. Conversely, layers \texttt{up\_block.1} and \texttt{down\_block.1}, with higher resolutions ($\geq$32), contribute more to style generation.
Building on this observation, we apply a prior-informed initialization strategy and implicit rank constraint during the merging process. For instance, during merging, we enforce sparsity along with the implicit rank constraint.



\subsubsection{Prior-informed merger initialization}
Instead of initializing the merger masks ($m_c$ and $m_s$ in Fig.~\ref{fig:method}) in the rank dimension with all ones, we incorporate layer-specific information to enhance the merging process. To this end, we use the observation of content-style encoding in SDXL architecture made earlier. 
To capture the dominance of content and style in different layers, the masks for the merged LoRAs are initialized using content (\(T_{\text{content}}\)) and style (\(T_{\text{style}}\)) thresholds. In content-dominant layers, more ones are assigned to content merger $m_c$ than style merger $m_s$, while the reverse is applied in style-dominant layers. The process begins with a normalized random vector, using thresholds to determine priority for content or style based on their ranks, defaulting to all ones when ranks are equal. For other layers, LoRAs are initialized with all ones to adaptively balance content and style merging. As demonstrated in Tab.~\ref{tab:loss_ablation}, this adaptive initialization improves performance. The threshold ablation is shown in Tab.~\ref{tab:init_thres_ablation}, and the detailed algorithm is provided in the supplementary.




\begin{figure}[h]
    \centering
    \begin{minipage}{0.23\textwidth}
        \centering
        \includegraphics[width=\textwidth]{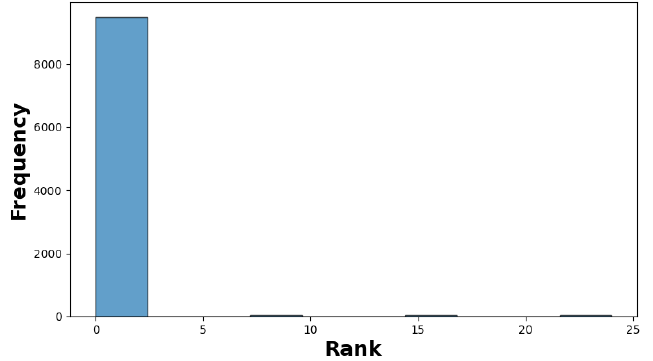}
        \subcaption{Histogram of rank across low resolution layers in style merger.}
    \end{minipage}%
    \hfill
    \begin{minipage}{0.23\textwidth}
        \centering
        \includegraphics[width=\textwidth]{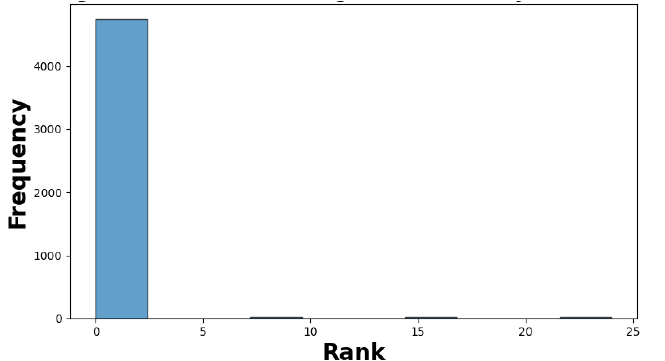}
        \subcaption{Histogram of rank across high resolution layers in content merger.}
    \end{minipage}
    \vspace{0.5cm} 
    \begin{minipage}{0.23\textwidth}
        \centering
        \includegraphics[width=\textwidth]{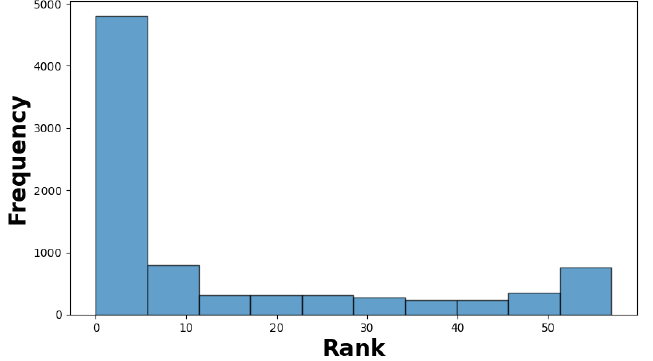}
        \subcaption{Histogram of rank across low resolution layers in content merger.}
    \end{minipage}%
    \hfill
    \begin{minipage}{0.23\textwidth}
        \centering
        \includegraphics[width=\textwidth]{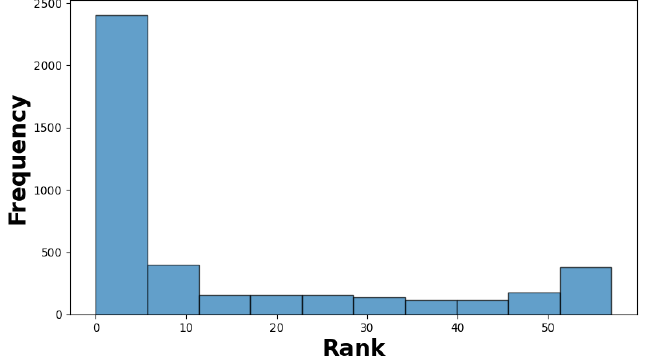}
        \subcaption{Histogram of rank across high resolution layers in style merger.}
    \end{minipage}
    \vspace{-0.5cm}
    \caption{Rank analysis. We plot the frequency of ranks across low-resolution layers (i.e., \texttt{up\_block.2}, \texttt{down\_block.2}, and \texttt{mid\_block} with a resolution < 32) and high-resolution layers (resolution >= 32) in content and style mergers post-training. The rank distribution confirms that in SDXL UNet architecture, for low-resolution layers, $Rank(m_c) > Rank(m_s)$, while for high-resolution layers, $Rank(m_s) > Rank(m_c)$.}
    \vspace{-0.4cm}
    \label{fig:rank_analysis}
\end{figure}



\begin{algorithm}
\SetAlgoLined
\KwIn{SDXL-UNet, content merger ($m_c$), style merger ($m_s$)}
\KwOut{Loss $\mathcal{L}_{\text{layer\_prior}}$ for merging layers}

\uIf{Resolution(SDXL-UNet layer) $<$ 32}{
    \tcp*[h]{$Rank(m_c) > Rank(m_s)$} \\
        $\mathcal{L}_{\text{layer\_prior}} \gets ||m_c||_1 + \lambda \, max\left(0, \|m_s\|_* - \|m_c\|_* \right) $\;
    }
\Else{
    { 
    \tcp*[h]{$Rank(m_s) > Rank(m_c)$} \\
    $\mathcal{L}_{\text{layer\_prior}} \gets ||m_s||_1 + \lambda \, max\left(0, \|m_c\|_* - \|m_s\|_* \right)$ \;
    }
}
\caption{Rank-constrained Layer-prior Merging}
\label{alg:layer_prior_merging}
\end{algorithm}

\subsubsection{Prior-informed rank constraint}
Based on our observations in Sec.~\ref{subsec:understanding_sdxl}, we propose to use the explicit layer-specific rank constraints during merging.
In Fig.~\ref{fig:rank_analysis}, we plot the frequency of ranks across low (i.e., \texttt{up\_block.2}, \texttt{down\_block.2}, and \texttt{mid\_block} with < 32 resolution) and high-resolution (>=32 resolution) layers in content and style mergers post training. Through the distribution of the ranks, we can verify that for low resolution layers, $Rank(m_c) > Rank(m_s)$, and for high resolution layers, $Rank(m_s) > Rank(m_c)$ holds in SDXL UNet architecture.

Therefore, for the content-dominant layers, we apply the constraint \( \text{Rank($m_c$)} \geq \text{Rank($m_s$)} \) during merging, and for the style-dominant layers, we reverse this constraint. In Lemma~\ref{Lemma:lemma_1}, we demonstrate that this rank constraint simplifies to a nuclear norm minimization problem under a sparsity constraint.

\begin{lemma}
\label{Lemma:lemma_1}
Let \( m_c \in \mathbb{R}^{m \times n} \) be a matrix representing the content merger and \( m_s \in \mathbb{R}^{m \times n} \) be a matrix representing the style merger. The problem of minimizing the \( l_1 \)-norm of \( m_c \) subject to a rank constraint on \( m_c \) can be written as:
\[
\min \|m_c\|_1 \quad \text{subject to} \quad \text{rank}(m_c) > \text{rank}(m_s)
\]
This problem is non-convex due to the rank constraint. A convex relaxation can be achieved by approximating the rank of a matrix using the nuclear norm \( \| \cdot \|_* \), which is the sum of the singular values of the matrix. Thus, the original problem can be relaxed to:
\[
\min \|m_c\|_1 \quad \text{subject to} \quad \|m_c\|_* > \|m_s\|_*
\]
where \( \|m_c\|_* \) denotes the nuclear norm of \( m_c \), and \( \|m_s\|_* \) is the nuclear norm of \( m_s \).

This relaxed problem can be approached via a Lagrangian penalty formulation:
\[
\mathcal{L}(m_c, m_s, \lambda) = \|m_c\|_1 + \lambda \max(0, \|m_s\|_* - \|m_c\|_*)
\]
for some penalty parameter \( \lambda \geq 0 \), which enforces the constraint \( \|m_c\|_* > \|m_s\|_* \) in the limit as \( \lambda \to \infty \).
\end{lemma}

\begin{proof}
    Proof is provided in the supplementary material.
\end{proof}

Consequently, in the content-dominant layers (i.e., \texttt{up\_block.2}, \texttt{down\_block.2}, and \texttt{mid\_block}), LoRAs are merged using the layer prior loss \( \mathcal{L}_{\text{layer\_prior}} \).
\begin{equation}
\mathcal{L}_{\text{layer\_prior}} = \|m_c\|_1 + \lambda \max(0, \|m_s\|_* - \|m_c\|_*)
\end{equation}


where \( m_s \) and \( m_c \) represent the style and content mergers, respectively.
Similarly, for the style-dominant layers, the process is reversed, as outlined in Algorithm~\ref{alg:layer_prior_merging}.


\subsection{Constyle loss : Cycle-consistent merging}

To further improve the content-style alignment, we introduce Constyle loss, leveraging the cycle-consistency between content and style.
Inspired by CycleGAN~\cite{zhu2017unpaired}, where content and style are transformed across domains with minimized reconstruction loss, we add and then remove style from content, ensuring minimal reconstruction error during LoRA merging. Similarly, content is added and removed from style images, as illustrated in Fig.~\ref{fig:method}.
Both processes provide feedback on the required blending while upholding rank constraints. Next, we describe cycle-consistent merging in detail (color coded for better explanation, refer to Fig.~\ref{fig:method}).

The cycle-consistent merging is performed as follows:

 \begin{enumerate}[leftmargin=0cm]  
    \item \textbf{Generate Individual LoRAs:}
    \begin{itemize}[leftmargin=0cm] 
        \item \textbf{Content LoRA:} Given a set of content images \( I_c \), we learn a content LoRA $\textcolor{blue}{\bf{L_c}}$, and tokens \textcolor{blue}{$\textbf{<V1>}$}, \textcolor{blue}{$\textbf{<S1>}$} using the prompt 
        $\textcolor{blue}{\bf{p_c}} = \texttt{"a \textcolor{blue}{$\textbf{<V1>}$} object in \textcolor{blue}{$\textbf{<S1>}$} style"}$.
        \item \textbf{Style LoRA:} For a (small) set of style images \( I_s \), we learn a style LoRA $\textcolor{red}{\bf{L_s}}$, and tokens \textcolor{red}{$\textbf{<V2>}$}, \textcolor{red}{$\textbf{<S2>}$} using the prompt 
        $\textcolor{red}{\bf{p_s}} = \texttt{"a \textcolor{red}{$\textbf{<V2>}$} object in \textcolor{red}{$\textbf{<S2>}$} style"}$.
    \end{itemize}
    
    \item \textbf{Cycle-Consistency Across Style:}
    Starting with a content image \( I_c \), we add noise to create its noisy latent. During the denoising process, the style prompt \texttt{"a \textcolor{blue}{$\textbf{<V1>}$} object in \textcolor{red}{$\textbf{<S2>}$} style"} is injected via the style LoRA $\textcolor{red}{\bf{L_s}}$ to produce a style-infused image \( I_{cs} \).
    This image \( I_{cs} \) is then re-noised and denoised with the content prompt $\textcolor{blue}{\bf{p_c}}$ (via $\textcolor{red}{\bf{L_s}}$) to remove style, resulting in the reconstructed content image \( I_{csc} \).
    To mitigate mode collapse from limited examples, we also generate a variant \( I_{cc} \) by denoising \( I_c \) solely with the content LoRA $\textcolor{blue}{\bf{L_c}}$ using $\textcolor{blue}{\bf{p_c}}$.
    We enforce cycle-consistency in style by minimizing the loss (Fig.~\ref{fig:method}):
    \vspace{-0.2cm}
        \[
        \mathcal{L}_{\text{cycle\_style}} = \text{MSE}(I_{cc}, I_{csc})
        \]
    
    \item \textbf{Cycle-Consistency Across Object (Content):} 
    In a similar manner, starting with a style image \( I_s \), we add noise and then denoise using a content prompt \texttt{"a \textcolor{blue}{$\textbf{<V1>}$} object in \textcolor{red}{$\textbf{<S2>}$} style"} (via the content LoRA $\textcolor{blue}{\bf{L_c}}$) to obtain a content-injected style image \( I_{sc} \).
    Next, we remove the injected object by denoising \( I_{sc} \) with the style prompt $\textcolor{red}{\bf{p_s}}$ via the content LoRA $\textcolor{blue}{\bf{L_c}}$, yielding \( I_{scs} \).
    Additionally, we generate a variant \( I_{ss} \) by denoising \( I_s \) with the style LoRA $\textcolor{red}{\bf{L_s}}$ using $\textcolor{red}{\bf{p_s}}$.
    The cycle-consistency in content is enforced by minimizing (Fig.~\ref{fig:method}):
    \vspace{-0.2cm}
    \[
        \mathcal{L}_{\text{cycle\_content}} = \text{MSE}(I_{ss}, I_{scs})
    \]
\end{enumerate}

The merged LoRA \( L_m \) is then trained by minimizing the following loss:
\vspace{-0.2cm}
\begin{align*}
    \mathcal{L}_{\text{constyle}} &= \left\|(D + L_m)(I_c, \textcolor{blue}{\bf{p_c}}) - (D + \textcolor{blue}{\bf{L_c}})(I_c, \textcolor{blue}{\bf{p_c}})\right\| \\
    &\quad + \left\|(D + L_m)(I_s, \textcolor{red}{\bf{p_s}}) - (D + \textcolor{red}{\bf{L_s}})(I_s, \textcolor{red}{\bf{p_s}})\right\| \\
    &\quad + \left( \mathcal{L}_{\text{cycle\_style}} + \mathcal{L}_{\text{cycle\_content}} \right)
\end{align*}
\vspace{-0.1cm}
where \( D \) is the text-to-image diffusion model and \( \lambda_{cycle} \) (set to 0.1) controls the weight of the cycle-consistency losses. The overall loss used during merging is given by
\vspace{-0.2cm}
\[
\mathcal{L} = \lambda_{\text{layer\_prior}}\,\mathcal{L}_{\text{layer\_prior}} + \lambda_{cycle}\,\mathcal{L}_{\text{constyle}}
\]
\vspace{-0.1cm}
which is further used to train the layer masks. During inference, the merged tokens are provided in the prompt (e.g., \texttt{"a \textcolor{blue}{$\textbf{<V1>}$} object in \textcolor{red}{$\textbf{<S2>}$} style running"}) to the diffusion model with merged LoRA \( L_m \), generating images that blend the desired content and style. Detailed algorithm is provided in the supplementary.

\subsection{Multi-concept stylization}

We further extend our approach to handle multi-concept stylization.
Given concepts $C_1$, .. $C_n$ and style $S$, we first merge individual concept-style LoRAs, $<C_1, S>, .. <C_n, S>$, using DuoLoRA/baselines and then combine them via naive merging: $C_{1,2,.,n,S}=\alpha_1<C_1, S>+..+\alpha_n<C_n, S>$, where $\alpha_{i}=1/n$. During inference, we use directional prompting (i.e., object in left/right etc) with the merged LoRA \( C_{1,2,.,n,S} \), using prompt \texttt{p = ``a $<C1>$ object on the left and a $<C2>$ object on the right .. in $<S>$ style''}. In this way, we extend DuoLoRA for multiple concepts.
\section{Experiments}
\label{sec:experiments}

\begin{table*}[h]
\centering
\captionsetup{font=footnotesize}
\caption{Performance comparison of content and style merging across different datasets and methods}
\vspace{-0.2cm}
\scalebox{0.55}{
\begin{tabular}{l|cccc|cccc|cccc|cccc|c|c|c}
\hline
\textbf{Method} & \multicolumn{4}{c|}{\textbf{Dreambooth + StyleDrop}} & \multicolumn{4}{c|}{\textbf{Subjectplop}} & \multicolumn{4}{c|}{\textbf{Subjectplop + StyleDrop}} & \multicolumn{4}{c|}{\textbf{Custom101 + StyleDrop}} & \textbf{\# Params} & \textbf{Param} & \textbf{Training}\\
 & \textbf{DINO} & \textbf{CLIP-I} & \textbf{CLIP-T} & \textbf{CSD-s} & \textbf{DINO} & \textbf{CLIP-I} & \textbf{CLIP-T} & \textbf{CSD-s} & \textbf{DINO} & \textbf{CLIP-I} & \textbf{CLIP-T} & \textbf{CSD-s} & \textbf{DINO} & \textbf{CLIP-I} & \textbf{CLIP-T} & \textbf{CSD-s} & \textbf{(M)} &\textbf{storage (MB)} & \textbf{time (m)}\\
\hline
Naïve Merging                 & 0.47 & 0.64 & 0.266 & 0.44 & 0.48 & 0.59 & 0.263 & 0.30 & 0.42 & 0.49 & 0.274 & 0.12 & 0.40 & 0.39 & 0.204 & 0.18 & - & - & -  \\
B-LoRA~\cite{frenkel2024implicit} (ECCV'24) & 0.45 & 0.57 & 0.281 & 0.28 & 0.64 & 0.57 & 0.275 & 0.32 & 0.63 & 0.56 & 0.281 & 0.14 & 0.49 & 0.51 & 0.263 & 0.25 & - & - & - \\
ZipLoRA~\cite{shah2025ziplora} (ECCV'24)    & 0.53 & 0.65 & 0.285 & 0.41 & 0.75 & 0.62 & 0.288 & 0.35 & 0.87 & 0.56 & 0.289 & 0.16 & 0.54 & 0.58 & 0.286 & 0.30 &  1.33 & 6.5 & 5.48\\
\rowcolor{gray!20} ZipRank                  & 0.53 & 0.64 & 0.287 & 0.42 & 0.71 & 0.62 & 0.291 & 0.35 & 0.86 & 0.56 & 0.296 & 0.17 & 0.56 & 0.58 & 0.295 & 0.32 & 0.07 & 0.35 & \textbf{5.28} \\
\rowcolor{gray!20} ZipRank + Layer-Priors & 0.54 & 0.67 & 0.293 & 0.45 & 0.73 & 0.63 & 0.310 & 0.37 & 0.90 & 0.56 & 0.302 & 0.18 & 0.59 & 0.60 & 0.307 & 0.35 & 0.07 & 0.35 & 6.02 \\
\rowcolor{gray!20} DuoLoRA & \textbf{0.56} & \textbf{0.69} & \textbf{0.314}  & \textbf{0.48} & \textbf{0.78} & \textbf{0.65} & \textbf{0.318} & \textbf{0.40} & \textbf{0.90} & \textbf{0.58} & \textbf{0.319} & \textbf{0.20} & \textbf{0.61} & \textbf{0.62} & \textbf{0.316} & \textbf{0.37} &  \textbf{0.07} & \textbf{0.35} & 6.38 \\
\hline
\end{tabular}}
\vspace{-0.4cm}
\label{tab:baseline_comparison}
\end{table*}

\noindent\textbf{Dataset.} 
We experiment on four datasets (datasets and splits provided in supplementary).

\noindent\textbf{(1) Dreambooth-SyleDrop}: We choose diverse set of content images from the Dreambooth dataset~\cite{ruiz2023dreambooth}, which consists of images of 30 subjects with 4-5 images per subjects. Style images are chosen from StyleDrop~\cite{sohn2023styledrop} dataset, where a single image is used per style.

\noindent\textbf{(2) Subjectplop}: subjectplop~\cite{ruiz2024magic} contains a single image for both content and style.

\noindent\textbf{(3) Subjectplop-StyleDrop}: We also benchmark on cross dataset, i.e., contents are taken from subjectplop dataset~\cite{ruiz2024magic}, and style images are taken from StyleDrop~\cite{sohn2023styledrop} dataset. Note here also, the content and style contains a single image.

\noindent\textbf{(4) Custom101-StyleDrop}: We conduct experiments on the real-world object-centric Custom101 dataset~\cite{kumari2023multi} (with 101 objects like human faces and everyday items) using styles from the Styledrop dataset~\cite{sohn2023styledrop}.

\noindent\textbf{Metrics.}
For content similarity, we use DINO similarity score~\cite{ruiz2023dreambooth}, i.e., the average pairwise cosine similarity of DINO ViT-B/6 embeddings of the content and generated images. For style similarity, we use the CLIP-I metric~\cite{ruiz2023dreambooth}, which is the average pairwise cosine similarity between CLIP embeddings of the style and generated images. 
Text similarity (CLIP-T~\cite{ruiz2023dreambooth}) is computed as the average cosine similarity between CLIP’s text and prompt embeddings.
We also report the Contrastive Style Descriptor-style (CSD-s) metric~\cite{somepalli2024measuring}, which is more appropriate to evaluate style similarity.

\noindent\textbf{Implementation details.}
In all our experiments, we have used SDXL v1.0 as our base model. To train the content and style LoRAs, we use Dreambooth finetuning with LoRA of rank 64. We update the LoRA weights using Adam optimizer for 1000 steps with batch size of 1 and learning rate of 5e-4. The text encoder of SDXL remains frozen during the LoRA finetuning. 

During the merging, we train the mergers in the rank dimension for 100 steps, with both the layer-prior loss ($\mathcal{L}_{layer\_prior}$) and the cycle loss ($\mathcal{L}_{constyle}$), with Adam optimizer and a learning rate of 0.01. The hyperparameters are choosen as follows, $\lambda_{cycle} = 0.01$ and $\lambda_{layer\_prior}=0.1$. 

\begin{figure}
\vspace{-0.2cm}
    \centering
        \centering
        \includegraphics[width=0.5\textwidth]{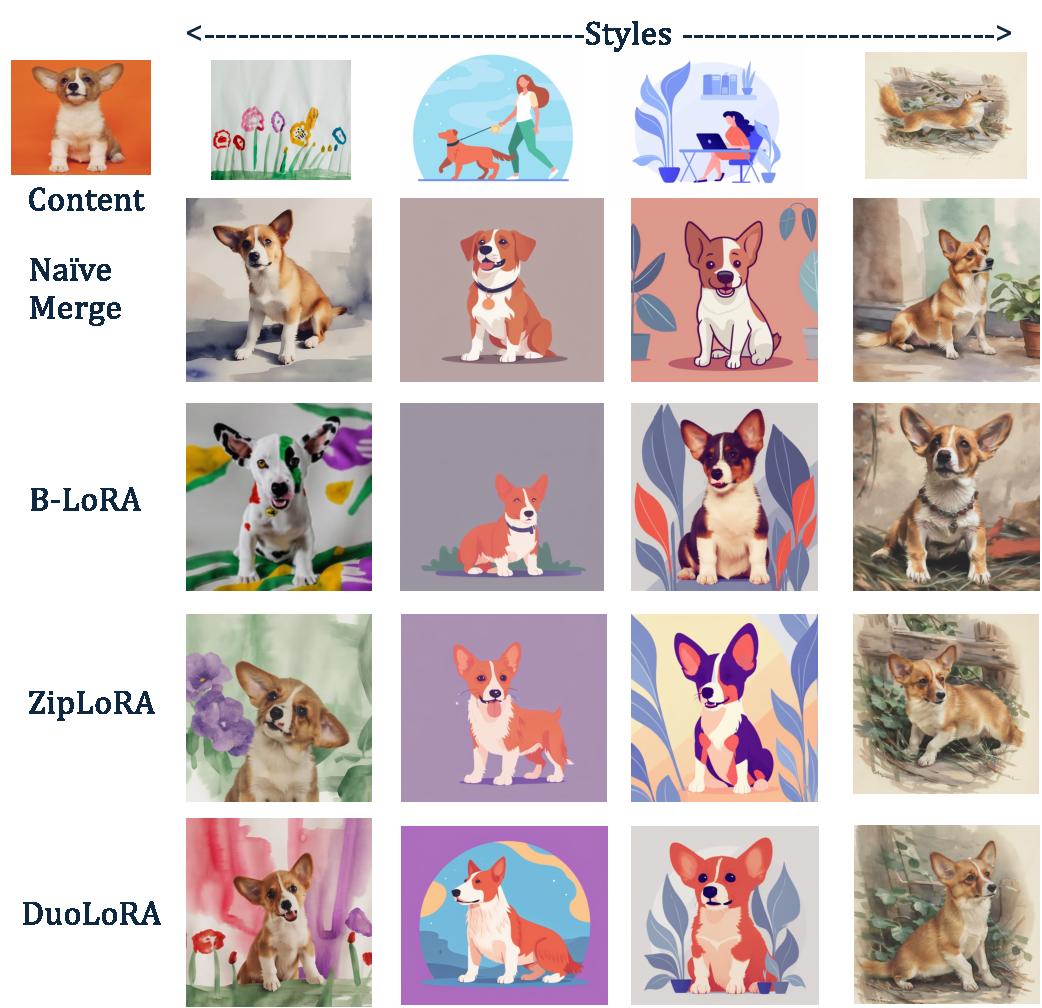}
        \vspace{-0.5cm}
        \caption{Qualitative Results on Dreambooth + StyleDrop.}
        \vspace{-0.8cm}
        \label{fig:qual_dog6}   
\end{figure}

\noindent\textbf{Baselines.}
For personalized stylization, we compare our method with the following baselines.
(1) Naive merging: The content and style adapters are added together with a blending scalar parameter during inference. This merging is training free. (2) B-LoRA~\cite{frenkel2024implicit}: Specific blocks are learnt for content and style in SDXL. (3) ZipLoRA~\cite{shah2025ziplora}: A weighting vector is learned at the output dimension of each adapter, such that their similarity is reduced. We used the parameters from ZipLoRA paper~\cite{shah2025ziplora}. and (4) Paircustomization~\cite{jones2024customizing}.

\noindent\textbf{Quantitative and Qualitative Results.}
We compare our method with Naive merging, B-LoRA and ziplora baselines. Tab.~\ref{tab:baseline_comparison} shows the comparison with different datasets and baseline methods. 
Our method outperforms the SOTA in all the DINO, CLIP and CSD metrics. 
Qualitative results are shown in Fig.~\ref{fig:qual_dog6}, Fig.~\ref{fig:custom101} and Fig.~\ref{fig:supple_qual_3}. Visually it is also evident that our method outperform the baselines. We also compare DuoLoRA with Pair-customization~\cite{jones2024customizing} in Tab.~\ref{tab:compare_paircustomization}, using their test set for a fair evaluation since their approach requires paired object and stylized images (details and qualitative results in supplementary).
To demonstrate that DuoLoRA generated images are not biased toward the content of the style images, we measured the DINO similarity between the content, style, and generated images. On the Dreambooth-StyleDrop benchmark, the average DINO similarity is 0.56 w.r.t content images and 0.25 w.r.t style images, indicating minimal content bias from the style images.

\begin{figure}
    \hspace{-0.3cm}
    \centering
    \includegraphics[scale=0.18]{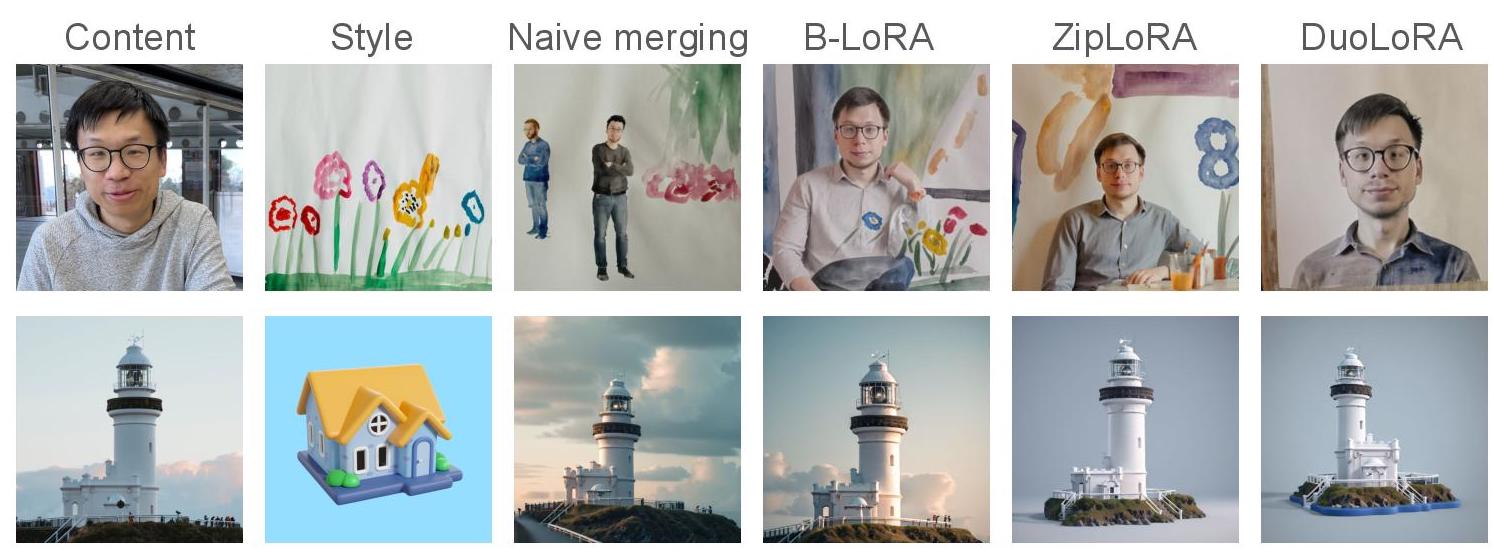}
    \captionsetup{font=footnotesize}
     \vspace{-0.6cm}
    \caption{\footnotesize{Qualitative Results on Custom101 (best viewed in color).}}
    \label{fig:custom101}
    \vspace{-0.4cm}
\end{figure}

\begin{figure}
    \centering
        \centering
        \includegraphics[width=0.5\textwidth]{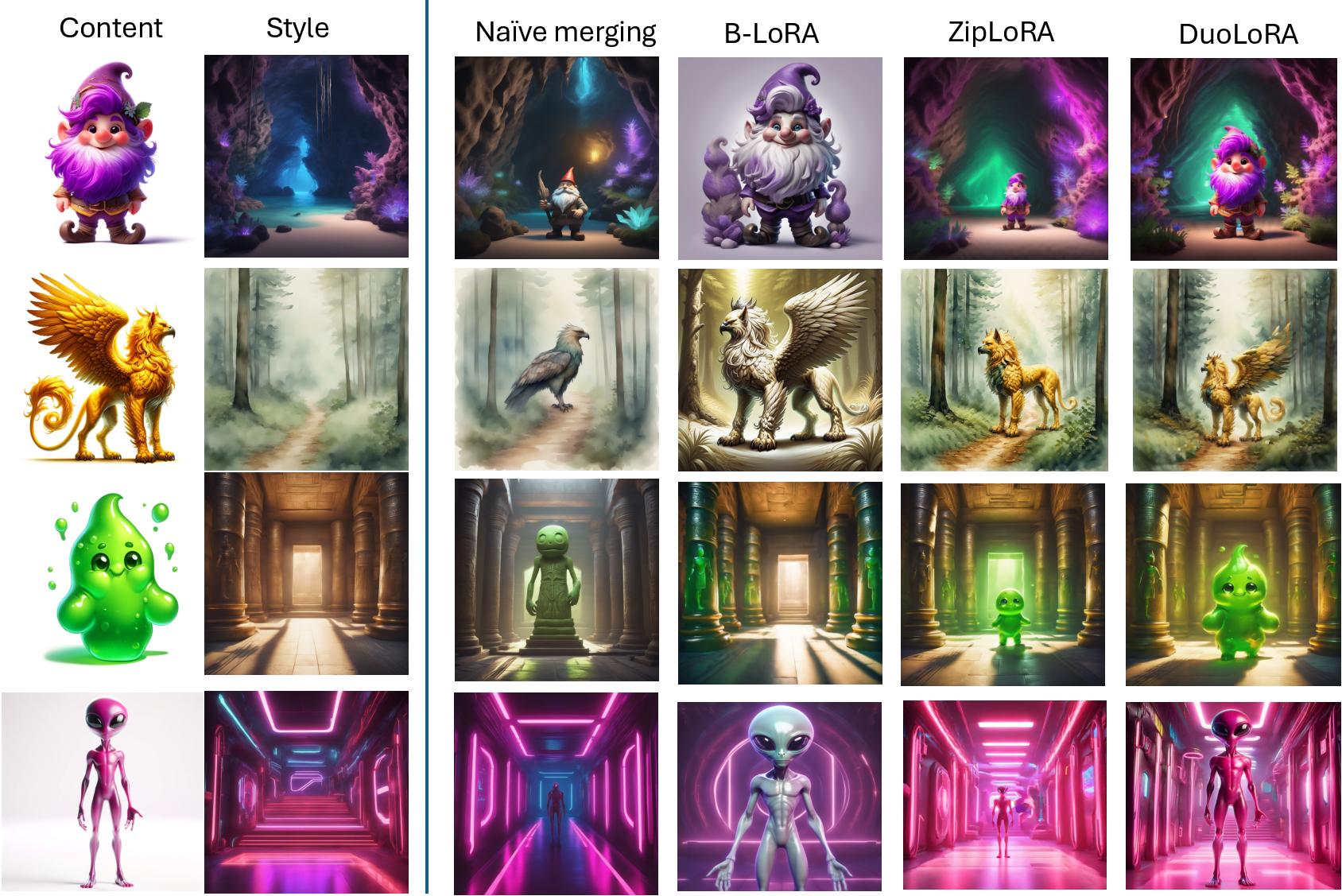} 
        \vspace{-0.6cm}
        \caption{Qualitative Results on Subjectplop (best viewed in color).}
        \vspace{-0.4cm}
        \label{fig:supple_qual_3}   
\end{figure}



\begin{figure}
    \centering
        \centering
        \includegraphics[width=0.45\textwidth]{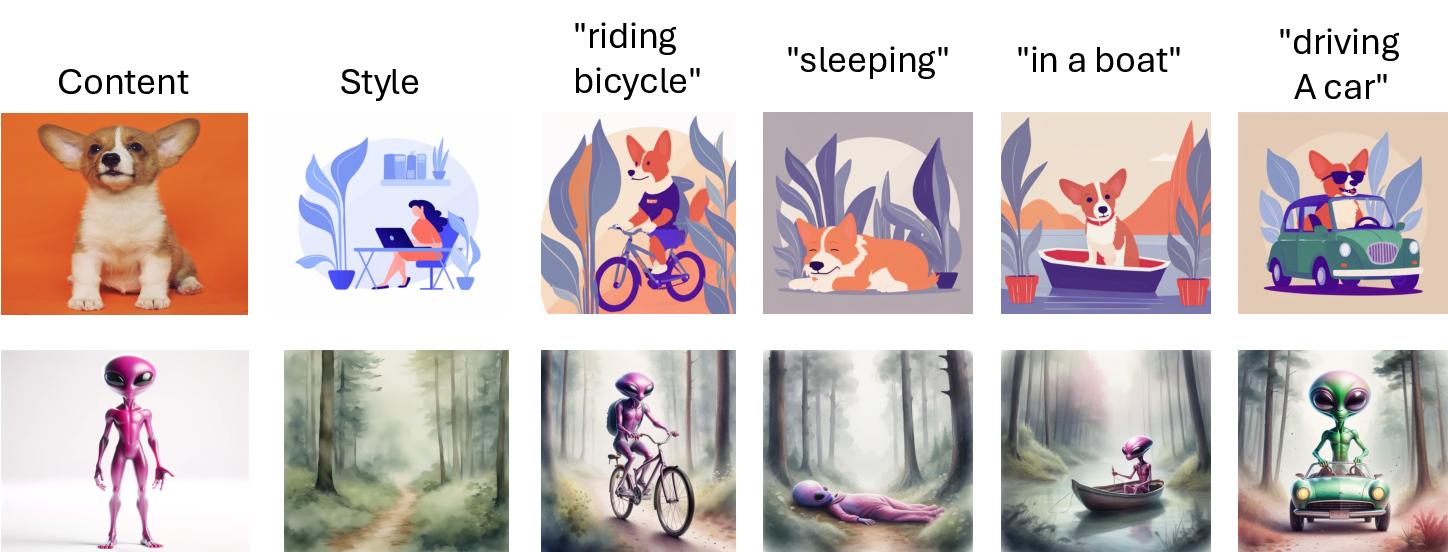} 
        \vspace{-0.3cm}
        \caption{Recontextualization through prompts.}
        \vspace{-0.6cm}
        \label{fig:recontextualization}   
\end{figure}

\noindent\textbf{Recontextulization.}
Tab.~\ref{tab:baseline_comparison} and Fig.~\ref{fig:recontextualization} show DuoLoRA successfully recontextualizes through text prompts, measured by CLIP-T, (e.g., ``a <V> object in <S> style riding a bicycle'') while seamlessly blending content and style (more in supplementary). Fig.~\ref{fig:recontextualization} also demonstrates that our blending preserves the model’s finer text-based editing capabilities, such as depicting ``sleeping'' through closed eyes.


\begin{table}[ht]
\centering
\captionsetup{font=footnotesize}
\vspace{-0.3cm}
\begin{minipage}[t]{0.2\textwidth}
\centering
\caption{Comparison with ~\cite{jones2024customizing}}
\vspace{-0.3cm}
\scalebox{0.55}{%
\begin{tabular}{lccc}
\hline
\textbf{Method} & \textbf{DINO} & \textbf{CLIP-I} & \textbf{CSD-s} \\
\hline
Paircustomization~\cite{jones2024customizing}   & 0.56 & 0.65 & 0.47 \\
\rowcolor{gray!20} DuoLoRA & 0.62 & 0.69 & 0.50 \\
\hline
\end{tabular}%
}
\label{tab:compare_paircustomization}
\end{minipage} \hspace{0.05\textwidth}
\begin{minipage}[t]{0.2\textwidth}
\centering
\caption{Same \#params.}
\vspace{-0.3cm}
\scalebox{0.55}{%
\begin{tabular}{lcccc}
\hline
\textbf{Method} & \textbf{DINO} & \textbf{CLIP-I} & \textbf{CSD-s} & \textbf{\#Params} \\
\hline
ZipLoRA   & 0.53 & 0.65 & 0.41 & 1.33M \\
\rowcolor{gray!20} DuoLoRA & 0.57 & 0.73 & 0.50 & 1.33M \\
\hline
\end{tabular}%
}
\label{tab:compare_ziplora_lp}
\end{minipage}
\vspace{-0.3cm}
\end{table}

\begin{figure}
    \centering
        \centering
        \includegraphics[width=0.5\textwidth]{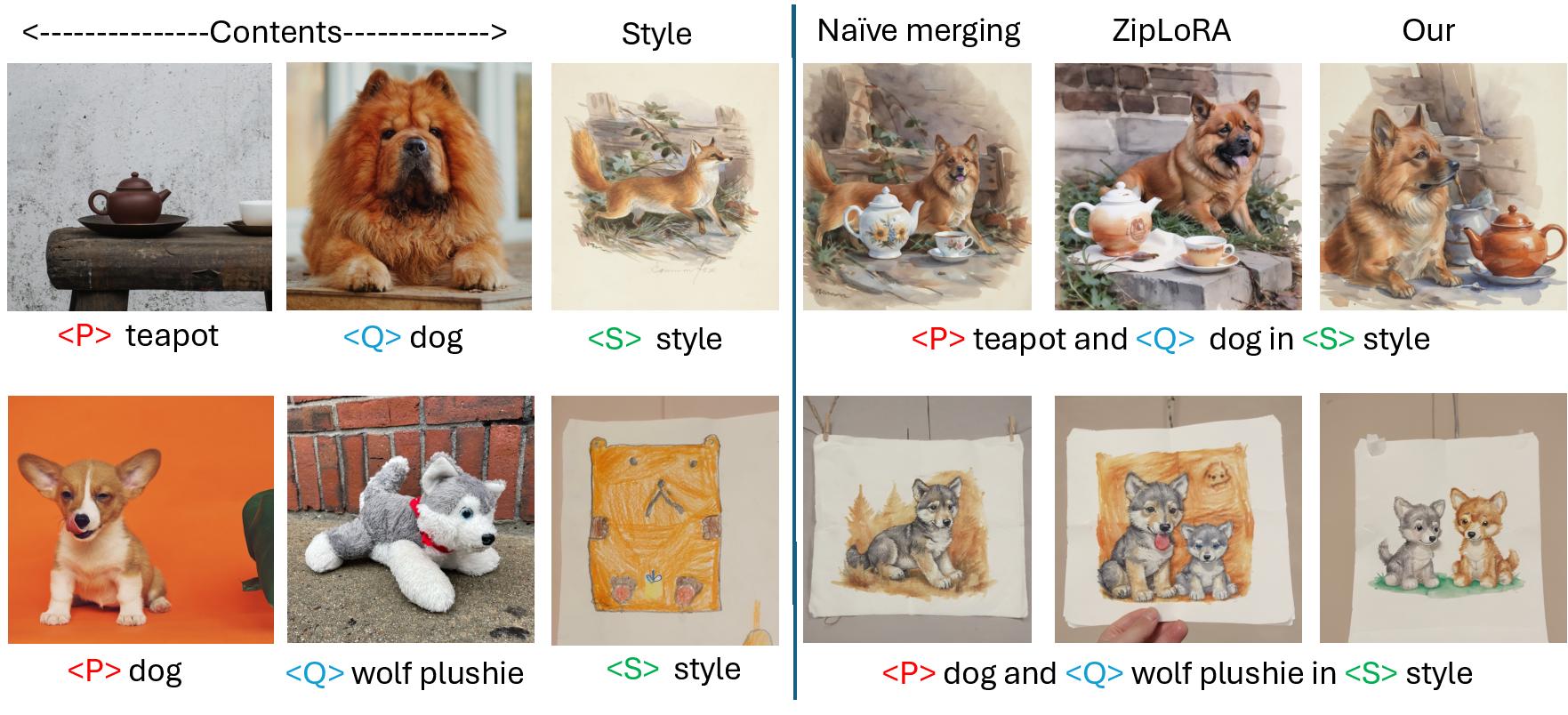} 
        \vspace{-0.8cm}
        \caption{Multi-concept stylization comparison (2 concepts).}
        \vspace{-0.2cm}
        \label{fig:multi_concept_2}   
\end{figure}

\noindent\textbf{Multi-concept stylization.}
The multi-concept stylization results are visualized in Fig.~\ref{fig:multi_concept_2},\ref{fig:multi_concept_4}, where our method shows qualitative improvements over the baselines in preserving both content and style details. For instance, as shown in Fig.~\ref{fig:multi_concept_2}, naive merging often loses details of either content, style, or both. Similarly, ZipLoRA struggles to retain content details, while our method successfully captures and maintains both content and style elements. We combine 2, 3, 4 concepts from Dreambooth dataset and styles from StyleDrop dataset and the results are shown in Tab.~\ref{tab:multi_concept}, Fig.~\ref{fig:multi_concept_2} and Fig.~\ref{fig:multi_concept_4}. Details and more results are in the supplementary.

\begin{table}
\centering
\vspace{-0.1cm}
\captionsetup{font=footnotesize}
\caption{\footnotesize Multi-concept comparison}
\vspace{-0.3cm}
\scalebox{0.55}{
\begin{tabular}{l|ccc|ccc|ccc}
\hline
\textbf{Method} & \multicolumn{3}{c|}{\textbf{2-concepts}} & \multicolumn{3}{c|}{\textbf{3-concepts}} & \multicolumn{3}{c}{\textbf{4-concepts}} \\
 & \textbf{DINO} & \textbf{CLIP-I} & \textbf{CSD-s} & \textbf{DINO} & \textbf{CLIP-I} & \textbf{CSD-s} & \textbf{DINO} & \textbf{CLIP-I} & \textbf{CSD-s}\\
\hline
Naïve Merging                 & 0.38 & 0.63 & 0.35 & 0.35 & 0.56 & 0.30 & 0.28 & 0.55 & 0.25\\
ZipLoRA   & 0.40 & 0.64 & 0.42 & 0.38 & 0.60 & 0.34 & 0.29 & 0.58 & 0.31\\
\rowcolor{gray!20} DuoLoRA                  & 0.45 & 0.66 & 0.47 & 0.40 & 0.64 & 0.39 & 0.32 & 0.63 & 0.35\\
\hline
\end{tabular}}
\label{tab:multi_concept}
\end{table}

\noindent\textbf{Other diffusion models.} 
We evaluate DuoLoRA on the SSD-1B~\cite{gupta2024progressive} and Segmind-Vega~\cite{gupta2024progressive} architectures, demonstrating its generalizability across architectures (Tab.~\ref{tab:arch_ablation}).

\noindent\textbf{Runtime and storage.}
We compare runtime and extra storage required for the mergers in DuoLoRA vs baselines in Tab.~\ref{tab:baseline_comparison} on an NVIDIA A6000 (24GB RAM). Our method slightly increase training time while reducing extra storage requirements compared to the baselines.

\noindent\textbf{User study.} 
Since perceptual metrics are not always reliable, we also conducted a human preference study using Amazon Mechanical Turk (AMT) for assesing the content-style alignment.
We asked 50 unbiased users rank our method against baselines (i.e., ``Naive merging'', ``B-LoRA'', ``ZipLoRA'', ``DuoLoRA'', ``None is satisfactory''), totaling 1000 questionnaires. The aggregate responses in Tab.~\ref{table:user_study} show that DuoLoRA generated images significantly outperformed the baselines by a large margin (50\%). Further details are provided in the Supplementary.

\begin{figure}
    \centering
    \vspace{-0.3cm}
    \includegraphics[scale=0.18]{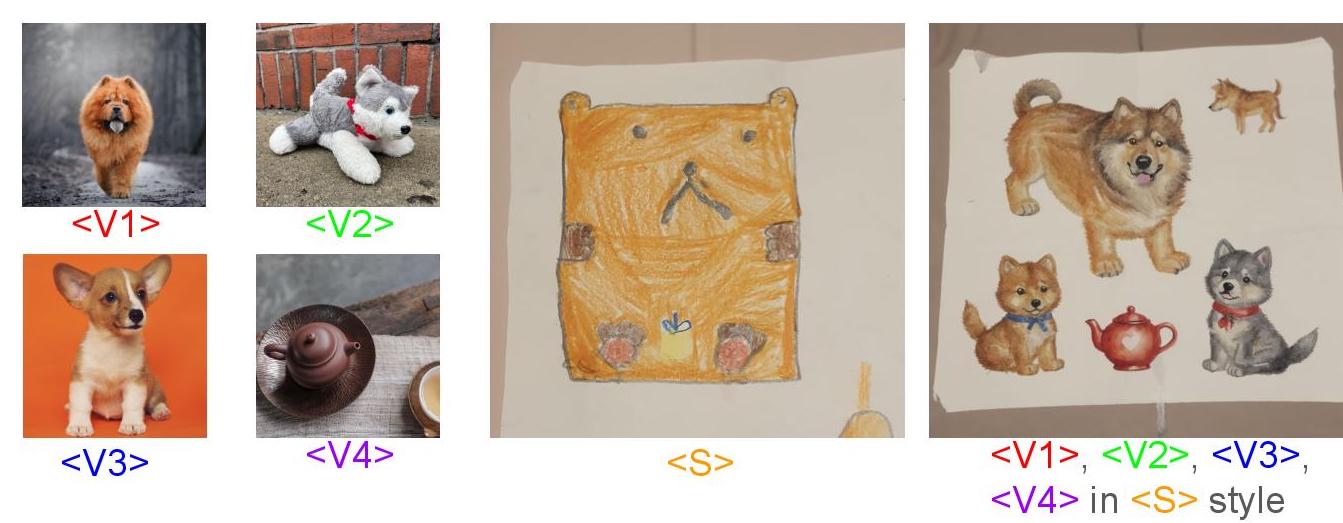}
    \vspace{-0.8cm}
    \captionsetup{font=footnotesize}
    \caption{\footnotesize{Multi-concept stylization (4 concepts).}}
    \label{fig:multi_concept_4}
    \vspace{-0.1cm}
\end{figure}

\begin{table}
\centering
\vspace{-0.2cm}
\captionsetup{font=footnotesize}
\caption{\footnotesize Comparison w.r.t architectures}
\vspace{-0.3cm}
\scalebox{0.55}{
\begin{tabular}{l|cccc|cccc}
\hline
\textbf{Method} & \multicolumn{4}{c|}{\textbf{SSD-1B}} & \multicolumn{4}{c}{\textbf{Segmind-Vega}}  \\
 & \textbf{DINO} & \textbf{CLIP-I} & \textbf{CSD-s} &  \textbf{\#params(M)} & \textbf{DINO} & \textbf{CLIP-I} & \textbf{CSD-s} & \textbf{\#params(M)}\\
\hline
Naïve Merging                 & 0.35 & 0.55 & 0.38 & - & 0.33 & 0.57 & 0.42 & -\\
ZipLoRA   & 0.42 & 0.65 & 0.46 & 0.61 & 0.40 & 0.67 & 0.50 & 0.26 \\
\rowcolor{gray!20} DuoLoRA                  & 0.44 & 0.71 & 0.52 & 0.03 & 0.41 & 0.72 & 0.56 & 0.02 \\
\hline
\end{tabular}}
\vspace{-0.5cm}
\label{tab:arch_ablation}
\end{table}

\begin{table}[!h]
\captionsetup{font=footnotesize}
  \small
  \centering
  \vspace{-0.2cm}
  \caption{\small{User study}}
  \vspace{-0.3cm}
  \scalebox{0.75}{
  \begin{tabular}{lccccc}
  \toprule
  {\bf None} & {\bf Naive merging} & {\bf B-LoRA} & {\bf ZipLoRA} & {\bf DuoLoRA} \\ 
  \midrule
  {0.0\%} & {10\%} & {18\%} & {22\%} & {50\%} \\
  \bottomrule
  \end{tabular}}
  \vspace{-0.2cm}
  \label{table:user_study}
\end{table}

\begin{table}[h]
\captionsetup{font=footnotesize}
\centering
\vspace{-0.1cm}
\caption{Ablation of content and style threshold ($T_{content}$, $T_{style}$)}
\vspace{-0.3cm}
\scalebox{0.6}{
\begin{tabular}{lcccccc}
\hline
\textbf{Method} & \textbf{$\lambda_{layer\_prior}$} & \(\mathbf{T_{content}}\) & \(\mathbf{T_{style}}\) & \textbf{DINO} & \textbf{CLIP-I}  & \textbf{CSD-s} \\
\hline
ZipRank                 & 0.1  & -   & -   & 0.53 & 0.64 & 0.42 \\
ZipRank + Init          & 0.1  & 0.1 & 0.0 & 0.56 & 0.65 & 0.43 \\
ZipRank + Init          & 0    & 0.75 & 0.50 & 0.57 & 0.64 & 0.41 \\
ZipRank + Init          & 0    & 1   & 0.75 & 0.53 & 0.64 & 0.42 \\
ZipRank + Init          & 0.1  & 0.75 & 0.50 & 0.50 & 0.64  & 0.44 \\
ZipRank + Init          & 0.1  & 1   & 0.75 & 0.47 & 0.64 & 0.43 \\
\hline
\end{tabular}}
\vspace{-0.2cm}
\label{tab:init_thres_ablation}
\end{table}

\begin{table}[h]
\vspace{-0.1cm}
\captionsetup{font=footnotesize}
\centering
\caption{Ablations of components}
\vspace{-0.3cm}
\scalebox{0.5}{
\begin{tabular}{lccccccccc}
\toprule
\textbf{ZipRank} & \textbf{Merger} & \textbf{Nuclear Norm} & \textbf{Sparsity} & \textbf{Cycle Loss} & \textbf{Cycle Loss} & \textbf{DINO} & \textbf{CLIP-I} & \textbf{CSD-s}  \\
{} & \textbf{Initialization} & \textbf{loss} & \textbf{loss} & \textbf{(Style)} & \textbf{(Content+Style)} & {} & {} & {}  \\
\midrule
{\ding{51}} & {\ding{55}} & {\ding{55}} & {\ding{55}} & {\ding{55}} & {\ding{55}} & 0.530 & 0.647 & 0.424  \\
{\ding{51}} & {\ding{51}} & {\ding{55}} & {\ding{55}} & {\ding{55}} & {\ding{55}} & 0.560 & 0.652 & 0.435  \\
{\ding{51}} & {\ding{51}} & {\ding{51}} & {\ding{55}} & {\ding{55}} & {\ding{55}} & 0.577 & 0.637 & 0.400  \\
{\ding{51}} & {\ding{51}} & {\ding{51}} & {\ding{55}} & {\ding{55}} & {\ding{55}} & 0.516 & 0.679 & 0.476  \\
{\ding{51}} & {\ding{51}} & {\ding{51}} & {\ding{51}} & {\ding{55}} & {\ding{55}} & 0.542 & 0.670 & 0.458 \\
{\ding{51}} & {\ding{51}} & {\ding{51}} & {\ding{51}} & {\ding{51}} & {\ding{55}} & 0.507 & 0.704 & 0.519 \\
{\ding{51}} & {\ding{51}} & {\ding{51}} & {\ding{51}} & {\ding{51}} & {\ding{51}} & 0.560 & 0.693 & 0.482 \\
\bottomrule
\end{tabular}}
\vspace{-0.2cm}
\label{tab:loss_ablation}
\end{table}

\begin{table}[ht]
\captionsetup{font=footnotesize}
\vspace{-0.1cm}
\caption{\footnotesize Ablation of parameters}
\vspace{-0.3cm}
\centering
\scalebox{0.5}{
\begin{tabular}{lccccccccc}
\hline
\textbf{Parameters} ($\lambda_{layer\_prior}$/$\lambda_{cycle}$) & 0/0 & 0/1 & 0/0.01 & 1/0 & 0.1/0 & 0.1/0.1 & 0.01/0.01 & 0.1/0.01 & 1.0/0.01 \\
\hline
\textbf{DINO}   & 0.53 & 0.45 & 0.47 & 0.57 & 0.55 & 0.52 & 0.48 & 0.51 & 0.56 \\
\textbf{CLIP-I} & 0.64 & 0.70 & 0.68 & 0.63 & 0.64 & 0.64 & 0.65 & 0.66 & 0.69\\
\textbf{CSD-s}  & 0.42 & 0.52 & 0.48 & 0.40 & 0.42 & 0.43 & 0.45 & 0.46 & 0.48\\
\hline
\end{tabular}}
\vspace{-0.3cm}
\label{tab:hyperparam_sensitivity}
\end{table}

\noindent\textbf{Ablations.}
We perform ablations on the loss and initialization when applied to content (Dreambooth dataset) and style (StyleDrop dataset) in Tab.~\ref{tab:loss_ablation}. Ablation of initialization parameters ($T_{content}$ and $T_{style}$) is also provided in Tab.~\ref{tab:init_thres_ablation}. We also present a hyperparameter sensitivity analysis on Dreambooth-StyleDrop in Tab.~\ref{tab:hyperparam_sensitivity} w.r.t loss coefficients ($\lambda_{layer\_prior}, \lambda_{cycle}$). With same number of parameters compared to baseline (ZipLoRA), DuoLoRA obtains better performance gains in Tab.~\ref{tab:compare_ziplora_lp}. More ablations are provided in the supplementary.

\noindent\textbf{Limitations.} 
At present, our approach handles only two concepts simultaneously, limiting its applicability for merging multiple concepts jointly. We aim to address this limitation in future work.

\vspace{-0.2cm}
\section{Conclusion}
\vspace{-0.2cm}
\label{sec:conclusion}
We investigate content and style personalization via LoRA merging, introducing DuoLoRA—a content-style personalization framework with three core components: (1) learning a mask in the rank dimension, (2) merging informed by layer priors, and (3) Constyle loss that leverages cycle-consistency between content and style. By learning the mask in the rank dimension, DuoLoRA enables adaptive rank flexibility with a significant reduction in trainable parameters (19x fewer). To further refine the merging process, we apply explicit rank constraints informed by layer priors and adaptive initialization. Additionally, we introduce Constyle loss, which uses content-style cycle-consistency to enhance merging. Experiments across multiple benchmarks show that our approach outperforms state-of-the-art methods.
\vspace{-0.2cm}
\section{Acknowledgement}
\vspace{-0.2cm}
Aniket Roy and RC acknowledge support through a fellowship from JHU + Amazon Initiative for Interactive AI (AI2AI) and ONR MURI grant N00014-20-1-2787. 
{
\small
\bibliographystyle{ieeenat_fullname}
\bibliography{main}

\begin{thebibliography}{45}
\providecommand{\natexlab}[1]{#1}
\providecommand{\url}[1]{\texttt{#1}}
\expandafter\ifx\csname urlstyle\endcsname\relax
  \providecommand{\doi}[1]{doi: #1}\else
  \providecommand{\doi}{doi: \begingroup \urlstyle{rm}\Url}\fi

\bibitem[Agarwal et~al.(2024)Agarwal, Karanam, and Srinivasan]{agarwal2024training}
Aishwarya Agarwal, Srikrishna Karanam, and Balaji~Vasan Srinivasan.
\newblock Training-free color-style disentanglement for constrained text-to-image synthesis.
\newblock \emph{arXiv preprint arXiv:2409.02429}, 2024.

\bibitem[Avrahami et~al.(2024)Avrahami, Gal, Chechik, Fried, Lischinski, Vahdat, and Nie]{avrahami2024diffuhaul}
Omri Avrahami, Rinon Gal, Gal Chechik, Ohad Fried, Dani Lischinski, Arash Vahdat, and Weili Nie.
\newblock Diffuhaul: A training-free method for object dragging in images.
\newblock In \emph{SIGGRAPH Asia 2024 Conference Papers}, pages 1--12, 2024.

\bibitem[Basu et~al.(2024)Basu, Rezaei, Kattakinda, Morariu, Zhao, Rossi, Manjunatha, and Feizi]{basu2024mechanistic}
Samyadeep Basu, Keivan Rezaei, Priyatham Kattakinda, Vlad~I Morariu, Nanxuan Zhao, Ryan~A Rossi, Varun Manjunatha, and Soheil Feizi.
\newblock On mechanistic knowledge localization in text-to-image generative models.
\newblock In \emph{Forty-first International Conference on Machine Learning}, 2024.

\bibitem[Bi et~al.(2024)Bi, Lu, Liu, Cun, Zhang, Li, and Xiao]{bi2024customttt}
Xiuli Bi, Jian Lu, Bo Liu, Xiaodong Cun, Yong Zhang, Weisheng Li, and Bin Xiao.
\newblock Customttt: Motion and appearance customized video generation via test-time training.
\newblock \emph{arXiv preprint arXiv:2412.15646}, 2024.

\bibitem[Borse et~al.(2025)Borse, Bhardwaj, Dastjerdi, Park, Kadambi, Shivakumar, Mandke, Nayak, Teague, Hayat, et~al.]{borse2025subzero}
Shubhankar Borse, Kartikeya Bhardwaj, Mohammad Reza~Karimi Dastjerdi, Hyojin Park, Shreya Kadambi, Shobitha Shivakumar, Prathamesh Mandke, Ankita Nayak, Harris Teague, Munawar Hayat, et~al.
\newblock Subzero: Composing subject, style, and action via zero-shot personalization.
\newblock \emph{arXiv preprint arXiv:2502.19673}, 2025.

\bibitem[Chang et~al.(2023)Chang, Zhang, Barber, Maschinot, Lezama, Jiang, Yang, Murphy, Freeman, Rubinstein, et~al.]{chang2023muse}
Huiwen Chang, Han Zhang, Jarred Barber, AJ Maschinot, Jose Lezama, Lu Jiang, Ming-Hsuan Yang, Kevin Murphy, William~T Freeman, Michael Rubinstein, et~al.
\newblock Muse: Text-to-image generation via masked generative transformers.
\newblock \emph{arXiv preprint arXiv:2301.00704}, 2023.

\bibitem[Choi et~al.(2024)Choi, Shin, Oh, Kim, and Yoon]{choi2024style}
Jooyoung Choi, Chaehun Shin, Yeongtak Oh, Heeseung Kim, and Sungroh Yoon.
\newblock Style-friendly snr sampler for style-driven generation.
\newblock \emph{arXiv preprint arXiv:2411.14793}, 2024.

\bibitem[Cohen et~al.(2024)Cohen, Nir, and Shamir]{cohen2024conditional}
Nadav~Z Cohen, Oron Nir, and Ariel Shamir.
\newblock Conditional balance: Improving multi-conditioning trade-offs in image generation.
\newblock \emph{arXiv preprint arXiv:2412.19853}, 2024.

\bibitem[Deng et~al.(2024)Deng, Zhou, Wang, and Mi]{deng2024magicstyle}
Zhaoli Deng, Kaibin Zhou, Fanyi Wang, and Zhenpeng Mi.
\newblock Magicstyle: Portrait stylization based on reference image.
\newblock \emph{arXiv preprint arXiv:2409.08156}, 2024.

\bibitem[Frenkel et~al.(2024)Frenkel, Vinker, Shamir, and Cohen-Or]{frenkel2024implicit}
Yarden Frenkel, Yael Vinker, Ariel Shamir, and Daniel Cohen-Or.
\newblock Implicit style-content separation using b-lora.
\newblock \emph{arXiv preprint arXiv:2403.14572}, 2024.

\bibitem[Gal et~al.(2022)Gal, Alaluf, Atzmon, Patashnik, Bermano, Chechik, and Cohen-Or]{gal2022image}
Rinon Gal, Yuval Alaluf, Yuval Atzmon, Or Patashnik, Amit~H Bermano, Gal Chechik, and Daniel Cohen-Or.
\newblock An image is worth one word: Personalizing text-to-image generation using textual inversion.
\newblock \emph{arXiv preprint arXiv:2208.01618}, 2022.

\bibitem[Gandikota et~al.(2023)Gandikota, Materzynska, Zhou, Torralba, and Bau]{gandikota2023concept}
Rohit Gandikota, Joanna Materzynska, Tingrui Zhou, Antonio Torralba, and David Bau.
\newblock Concept sliders: Lora adaptors for precise control in diffusion models.
\newblock \emph{arXiv preprint arXiv:2311.12092}, 2023.

\bibitem[Gatys et~al.(2016)Gatys, Ecker, and Bethge]{gatys2016image}
Leon~A Gatys, Alexander~S Ecker, and Matthias Bethge.
\newblock Image style transfer using convolutional neural networks.
\newblock In \emph{Proceedings of the IEEE conference on computer vision and pattern recognition}, pages 2414--2423, 2016.

\bibitem[Ge et~al.(2024)Ge, Liu, Fan, Jiang, Huang, Qin, Gu, Li, and Duan]{ge2024tuning}
Yanqi Ge, Jiaqi Liu, Qingnan Fan, Xi Jiang, Ye Huang, Shuai Qin, Hong Gu, Wen Li, and Lixin Duan.
\newblock Tuning-free adaptive style incorporation for structure-consistent text-driven style transfer.
\newblock \emph{arXiv preprint arXiv:2404.06835}, 2024.

\bibitem[Gupta et~al.(2024)Gupta, Jaddipal, Prabhala, Paul, and Von~Platen]{gupta2024progressive}
Yatharth Gupta, Vishnu~V Jaddipal, Harish Prabhala, Sayak Paul, and Patrick Von~Platen.
\newblock Progressive knowledge distillation of stable diffusion xl using layer level loss.
\newblock \emph{arXiv preprint arXiv:2401.02677}, 2024.

\bibitem[Hertz et~al.(2024)Hertz, Voynov, Fruchter, and Cohen-Or]{hertz2024style}
Amir Hertz, Andrey Voynov, Shlomi Fruchter, and Daniel Cohen-Or.
\newblock Style aligned image generation via shared attention.
\newblock In \emph{Proceedings of the IEEE/CVF Conference on Computer Vision and Pattern Recognition}, pages 4775--4785, 2024.

\bibitem[Hu et~al.(2021)Hu, Shen, Wallis, Allen-Zhu, Li, Wang, Wang, and Chen]{hu2021lora}
Edward~J Hu, Yelong Shen, Phillip Wallis, Zeyuan Allen-Zhu, Yuanzhi Li, Shean Wang, Lu Wang, and Weizhu Chen.
\newblock Lora: Low-rank adaptation of large language models.
\newblock \emph{arXiv preprint arXiv:2106.09685}, 2021.

\bibitem[Hu et~al.(2024)Hu, Xing, Zhang, and Yu]{hu2024vectorpainter}
Juncheng Hu, Ximing Xing, Jing Zhang, and Qian Yu.
\newblock Vectorpainter: Advanced stylized vector graphics synthesis using stroke-style priors.
\newblock \emph{arXiv preprint arXiv:2405.02962}, 2024.

\bibitem[Jiang and Chen(2024)]{jiang2024artist}
Ruixiang Jiang and Changwen Chen.
\newblock Artist: Aesthetically controllable text-driven stylization without training.
\newblock \emph{arXiv preprint arXiv:2407.15842}, 2024.

\bibitem[Jones et~al.(2024)Jones, Wang, Kumari, Bau, and Zhu]{jones2024customizing}
Maxwell Jones, Sheng-Yu Wang, Nupur Kumari, David Bau, and Jun-Yan Zhu.
\newblock Customizing text-to-image models with a single image pair.
\newblock \emph{arXiv preprint arXiv:2405.01536}, 2024.

\bibitem[Kompanowski and Hua(2024)]{kompanowski2024dream}
Hubert Kompanowski and Binh-Son Hua.
\newblock Dream-in-style: Text-to-3d generation using stylized score distillation.
\newblock \emph{arXiv preprint arXiv:2406.18581}, 2024.

\bibitem[Kumari et~al.(2023)Kumari, Zhang, Zhang, Shechtman, and Zhu]{kumari2023multi}
Nupur Kumari, Bingliang Zhang, Richard Zhang, Eli Shechtman, and Jun-Yan Zhu.
\newblock Multi-concept customization of text-to-image diffusion.
\newblock In \emph{Proceedings of the IEEE/CVF Conference on Computer Vision and Pattern Recognition}, pages 1931--1941, 2023.

\bibitem[Liu et~al.(2024)Liu, Shah, Cui, and Lazebnik]{liu2024unziplora}
Chang Liu, Viraj Shah, Aiyu Cui, and Svetlana Lazebnik.
\newblock Unziplora: Separating content and style from a single image.
\newblock \emph{arXiv preprint arXiv:2412.04465}, 2024.

\bibitem[Ohm et~al.(2025)Ohm, Karjus, Tamm, and Schich]{ohm2025fruit}
Tillmann Ohm, Andres Karjus, Mikhail~V Tamm, and Maximilian Schich.
\newblock fruit-salad: A style aligned artwork dataset to reveal similarity perception in image embeddings.
\newblock \emph{Scientific Data}, 12\penalty0 (1):\penalty0 254, 2025.

\bibitem[Ouyang et~al.(2025)Ouyang, Li, and Hou]{ouyang2025k}
Ziheng Ouyang, Zhen Li, and Qibin Hou.
\newblock K-lora: Unlocking training-free fusion of any subject and style loras.
\newblock \emph{arXiv preprint arXiv:2502.18461}, 2025.

\bibitem[Rout et~al.(2024)Rout, Chen, Ruiz, Kumar, Caramanis, Shakkottai, and Chu]{rout2024rb}
Litu Rout, Yujia Chen, Nataniel Ruiz, Abhishek Kumar, Constantine Caramanis, Sanjay Shakkottai, and Wen-Sheng Chu.
\newblock Rb-modulation: Training-free personalization of diffusion models using stochastic optimal control.
\newblock \emph{arXiv preprint arXiv:2405.17401}, 2024.

\bibitem[Ruiz et~al.(2023)Ruiz, Li, Jampani, Pritch, Rubinstein, and Aberman]{ruiz2023dreambooth}
Nataniel Ruiz, Yuanzhen Li, Varun Jampani, Yael Pritch, Michael Rubinstein, and Kfir Aberman.
\newblock Dreambooth: Fine tuning text-to-image diffusion models for subject-driven generation.
\newblock In \emph{Proceedings of the IEEE/CVF conference on computer vision and pattern recognition}, pages 22500--22510, 2023.

\bibitem[Ruiz et~al.(2024)Ruiz, Li, Wadhwa, Pritch, Rubinstein, Jacobs, and Fruchter]{ruiz2024magic}
Nataniel Ruiz, Yuanzhen Li, Neal Wadhwa, Yael Pritch, Michael Rubinstein, David~E Jacobs, and Shlomi Fruchter.
\newblock Magic insert: Style-aware drag-and-drop.
\newblock \emph{arXiv preprint arXiv:2407.02489}, 2024.

\bibitem[Shah et~al.(2025)Shah, Ruiz, Cole, Lu, Lazebnik, Li, and Jampani]{shah2025ziplora}
Viraj Shah, Nataniel Ruiz, Forrester Cole, Erika Lu, Svetlana Lazebnik, Yuanzhen Li, and Varun Jampani.
\newblock Ziplora: Any subject in any style by effectively merging loras.
\newblock In \emph{European Conference on Computer Vision}, pages 422--438. Springer, 2025.

\bibitem[Shenaj et~al.(2024)Shenaj, Bohdal, Ozay, Zanuttigh, and Michieli]{shenaj2024lora}
Donald Shenaj, Ondrej Bohdal, Mete Ozay, Pietro Zanuttigh, and Umberto Michieli.
\newblock Lora. rar: Learning to merge loras via hypernetworks for subject-style conditioned image generation.
\newblock \emph{arXiv preprint arXiv:2412.05148}, 2024.

\bibitem[Sohn et~al.(2023)Sohn, Ruiz, Lee, Chin, Blok, Chang, Barber, Jiang, Entis, Li, et~al.]{sohn2023styledrop}
Kihyuk Sohn, Nataniel Ruiz, Kimin Lee, Daniel~Castro Chin, Irina Blok, Huiwen Chang, Jarred Barber, Lu Jiang, Glenn Entis, Yuanzhen Li, et~al.
\newblock Styledrop: Text-to-image generation in any style.
\newblock \emph{arXiv preprint arXiv:2306.00983}, 2023.

\bibitem[Somepalli et~al.(2024)Somepalli, Gupta, Gupta, Palta, Goldblum, Geiping, Shrivastava, and Goldstein]{somepalli2024measuring}
Gowthami Somepalli, Anubhav Gupta, Kamal Gupta, Shramay Palta, Micah Goldblum, Jonas Geiping, Abhinav Shrivastava, and Tom Goldstein.
\newblock Measuring style similarity in diffusion models.
\newblock \emph{arXiv preprint arXiv:2404.01292}, 2024.

\bibitem[Song et~al.(2024)Song, Huang, Xie, Wang, and Wang]{song2024style3d}
Bingjie Song, Xin Huang, Ruting Xie, Xue Wang, and Qing Wang.
\newblock Style3d: Attention-guided multi-view style transfer for 3d object generation.
\newblock \emph{arXiv preprint arXiv:2412.03571}, 2024.

\bibitem[Stoica et~al.(2023)Stoica, Bolya, Bjorner, Ramesh, Hearn, and Hoffman]{stoica2023zipit}
George Stoica, Daniel Bolya, Jakob Bjorner, Pratik Ramesh, Taylor Hearn, and Judy Hoffman.
\newblock Zipit! merging models from different tasks without training.
\newblock \emph{arXiv preprint arXiv:2305.03053}, 2023.

\bibitem[SUN et~al.()SUN, Guo, Yang, Wang, et~al.]{sunsast}
XINYUE SUN, Jing Guo, Shuai Yang, Kai Wang, et~al.
\newblock Sast: Semantic-aware stylized text-to-image generation.
\newblock \emph{Jing and Yang, Shuai and Wang, Kai, Sast: Semantic-Aware Stylized Text-to-Image Generation}.

\bibitem[Wang et~al.(2024)Wang, Spinelli, Wang, Bai, Qin, and Chen]{wang2024instantstyle}
Haofan Wang, Matteo Spinelli, Qixun Wang, Xu Bai, Zekui Qin, and Anthony Chen.
\newblock Instantstyle: Free lunch towards style-preserving in text-to-image generation.
\newblock \emph{arXiv preprint arXiv:2404.02733}, 2024.

\bibitem[Xie et~al.(2024)Xie, Zhang, Tang, Wu, Chen, Li, and Jin]{xie2024styletex}
Zhiyu Xie, Yuqing Zhang, Xiangjun Tang, Yiqian Wu, Dehan Chen, Gongsheng Li, and Xiaogang Jin.
\newblock Styletex: Style image-guided texture generation for 3d models.
\newblock \emph{ACM Transactions on Graphics (TOG)}, 43\penalty0 (6):\penalty0 1--14, 2024.

\bibitem[Xing et~al.(2024)Xing, Wang, Sun, Wang, Bai, Ai, Huang, and Li]{xing2024csgo}
Peng Xing, Haofan Wang, Yanpeng Sun, Qixun Wang, Xu Bai, Hao Ai, Renyuan Huang, and Zechao Li.
\newblock Csgo: Content-style composition in text-to-image generation.
\newblock \emph{arXiv preprint arXiv:2408.16766}, 2024.

\bibitem[Xu et~al.(2024{\natexlab{a}})Xu, Tang, Cao, Zhang, Deussen, Dong, Li, and Lee]{xu2024break}
Yu Xu, Fan Tang, Juan Cao, Yuxin Zhang, Oliver Deussen, Weiming Dong, Jintao Li, and Tong-Yee Lee.
\newblock Break-for-make: Modular low-rank adaptations for composable content-style customization.
\newblock \emph{arXiv preprint arXiv:2403.19456}, 2024{\natexlab{a}}.

\bibitem[Xu et~al.(2024{\natexlab{b}})Xu, Wang, Xiao, Liu, and Chen]{xu2024freetuner}
Youcan Xu, Zhen Wang, Jun Xiao, Wei Liu, and Long Chen.
\newblock Freetuner: Any subject in any style with training-free diffusion.
\newblock \emph{arXiv preprint arXiv:2405.14201}, 2024{\natexlab{b}}.

\bibitem[Yang et~al.(2024)Yang, Wang, Peng, Song, Chen, Li, Yang, Lu, Cai, Wu, et~al.]{yang2024lora}
Yang Yang, Wen Wang, Liang Peng, Chaotian Song, Yao Chen, Hengjia Li, Xiaolong Yang, Qinglin Lu, Deng Cai, Boxi Wu, et~al.
\newblock Lora-composer: Leveraging low-rank adaptation for multi-concept customization in training-free diffusion models.
\newblock \emph{arXiv preprint arXiv:2403.11627}, 2024.

\bibitem[Ye et~al.(2023)Ye, Zhang, Liu, Han, and Yang]{ye2023ip}
Hu Ye, Jun Zhang, Sibo Liu, Xiao Han, and Wei Yang.
\newblock Ip-adapter: Text compatible image prompt adapter for text-to-image diffusion models.
\newblock \emph{arXiv preprint arXiv:2308.06721}, 2023.

\bibitem[Zhang et~al.(2024)Zhang, Sohn, Hahn, Shi, and Essa]{zhang2024finestyle}
Gong Zhang, Kihyuk Sohn, Meera Hahn, Humphrey Shi, and Irfan Essa.
\newblock Finestyle: Fine-grained controllable style personalization for text-to-image models.
\newblock \emph{Advances in Neural Information Processing Systems}, 37:\penalty0 52937--52961, 2024.

\bibitem[Zhang et~al.(2023)Zhang, Rao, and Agrawala]{zhang2023adding}
Lvmin Zhang, Anyi Rao, and Maneesh Agrawala.
\newblock Adding conditional control to text-to-image diffusion models.
\newblock In \emph{Proceedings of the IEEE/CVF International Conference on Computer Vision}, pages 3836--3847, 2023.

\bibitem[Zhu et~al.(2017)Zhu, Park, Isola, and Efros]{zhu2017unpaired}
Jun-Yan Zhu, Taesung Park, Phillip Isola, and Alexei~A Efros.
\newblock Unpaired image-to-image translation using cycle-consistent adversarial networks.
\newblock In \emph{Proceedings of the IEEE international conference on computer vision}, pages 2223--2232, 2017.

\end{thebibliography}
}
\clearpage
\setcounter{page}{1}
\maketitlesupplementary

In this supplementary material, we will provide the following details. 
\begin{enumerate}
    \item Training details.
    \item Theoretical analysis.
    \item Algorithm details.
    \item Details of user study.
    \item Additional results and ablations.
\end{enumerate}

\section{Training details.}
We have provided the hyperparameters for each of the datasets, i.e., Dreambooth + SyleDrop, Subjectplop, Subjectplop + SyleDrop and Custom101 + SyleDrop in Table.~\ref{tab:db_sd_hparams}, Table.~\ref{tab:sp_hparams}, Table.~\ref{tab:sp_sd_hparams} and Table.~\ref{tab:custom_hparams} respectively. 
For each of the dataset, the content and style images are provided in the supplementary as attachment.

For the joint training baseline, we are using Dreambooth~\cite{ruiz2023dreambooth} style joint training on SDXL with learning rate of 5.e-6. Training for 500 steps across all Unet parameters on a resolution of 768. 

\begin{table}[!h]
    \centering
    \caption{Hyperparameters for Dreambooth + SyleDrop}
    \scalebox{1.0}{
    \begin{tabular}{cc}
    \toprule
    Hyperparameter & Values \\
    \midrule
    $\lambda_{layer\_prior}$ & 0.1 \\
    $\lambda_{cycle}$ & 0.01 \\
    Base diffusion model & SDXL v1.0 \\
    LoRA rank & 64 \\
    Learning rate of LoRA & $5e^{-5}$\\
    Learning rate of mergers & $0.001$\\
    Batch size & 1\\
    resolution & 1024 \\
    $T_{content}$ & 0.1\\
    $T_{style}$ & 0.0 \\
    \bottomrule
    \end{tabular}}
    \label{tab:db_sd_hparams}
\end{table}

\begin{table}[!h]
    \centering
    \caption{Hyperparameters for Subjectplop}
    \scalebox{1.0}{
    \begin{tabular}{cc}
    \toprule
    Hyperparameter & Values \\
    \midrule
    $\lambda_{layer\_prior}$ & 0.01 \\
    $\lambda_{cycle}$ & 0.01 \\
    Base diffusion model & SDXL v1.0 \\
    LoRA rank & 64 \\
    Learning rate of LoRA & $5e^{-5}$\\
    Learning rate of mergers & $0.001$\\
    Batch size & 1\\
    resolution & 1024 \\
    $T_{content}$ & 0.75\\
    $T_{style}$ & 0.5 \\
    \bottomrule
    \end{tabular}}
    \label{tab:sp_hparams}
\end{table}

\begin{table}[!h]
    \centering
    \caption{Hyperparameters for Subjectplop + SyleDrop}
    \scalebox{1.0}{
    \begin{tabular}{cc}
    \toprule
    Hyperparameter & Values \\
    \midrule
    $\lambda_{layer\_prior}$ & 0.1 \\
    $\lambda_{cycle}$ & 0.01 \\
    Base diffusion model & SDXL v1.0 \\
    LoRA rank & 64 \\
    Learning rate of LoRA & $5e^{-5}$\\
    Learning rate of mergers & $0.001$\\
    Batch size & 1\\
    resolution & 1024 \\
    $T_{content}$ & 0.1\\
    $T_{style}$ & 0.0 \\
    \bottomrule
    \end{tabular}}
    \label{tab:sp_sd_hparams}
\end{table}

\begin{table}[!h]
    \centering
    \caption{Hyperparameters for Custom101 + SyleDrop}
    \scalebox{1.0}{
    \begin{tabular}{cc}
    \toprule
    Hyperparameter & Values \\
    \midrule
    $\lambda_{layer\_prior}$ & 0.1 \\
    $\lambda_{cycle}$ & 0.01 \\
    Base diffusion model & SDXL v1.0 \\
    LoRA rank & 64 \\
    Learning rate of LoRA & $5e^{-5}$\\
    Learning rate of mergers & $0.001$\\
    Batch size & 1\\
    resolution & 1024 \\
    $T_{content}$ & 0.1\\
    $T_{style}$ & 0.0 \\
    \bottomrule
    \end{tabular}}
    \label{tab:custom_hparams}
\end{table}


\section{Theoretical Analysis}
\label{sec:theory}

In this section, we provide the proof the theoretical results provided in the main paper.

\begin{theorem} In Low-Rank Adaptation (LoRA) merging, under the same parameter budget, the approximation error resulting from rank dimension masking is less than or equal to that from output dimension masking. Formally,
\[
E_{\text{rank}} \leq E_{\text{out}},
\]
where:
\[
E_{\text{rank}} = \| X - \Delta W_{\text{rank}} \|_F
\]
is the approximation error using rank dimension masking, and
\[
E_{\text{out}} = \| X - \Delta W_{\text{out}} \|_F
\]
is the approximation error using output dimension masking.
\end{theorem}

\begin{proof}
To compare the two masking strategies fairly, we ensure that both use the same number of parameters.
Parameter Count for Rank Dimension Masking, 
\[
P_{\text{rank}} = s (d_{\text{out}} + d_{\text{in}}).
\]
Parameter count for output dimension masking is,
\[
P_{\text{out}} = d_s \times r + r \times d_{\text{in}} = r (d_s + d_{\text{in}}).
\]

Now, Setting Equal Parameter Budgets,
Set $P_{\text{rank}} = P_{\text{out}}$:
\[
s (d_{\text{out}} + d_{\text{in}}) = r (d_s + d_{\text{in}}).
\]
Assuming $d_{\text{out}} = d_{\text{in}} = d$ for simplicity:
\[
s (2d) = r (d_s + d).
\]
Solving for $s$:
\[
s = \frac{r (d_s + d)}{2d}.
\]

Next, we compare the Approximation Errors.

Since we retain the $s$ largest singular values to minimize the error, the approximation error for Rank Dimension Masking is:
\[
E_{\text{rank}} = \left( \sum_{i = s+1}^{p} \sigma_i^2 \right)^{1/2}.
\]

The exact computation of Output Dimension Masking error ($E_{\text{out}}$) is complex due to the loss of orthogonality in $U_r$ caused by masking. However, we can establish a lower bound.

The total energy (sum of squares) in $U_r$ is:
\[
\| U_r \|_F^2 = \operatorname{trace}(U_r^\top U_r) = r.
\]
where each row of $U_r$ contributes equally on average to this total energy.

The fraction of rows masked out (i.e., energy removed by masking) is:
\[
f = \frac{d - d_s}{d}.
\]
Therefore, the approximate fraction of energy removed is $f$.

Next, we get the lower bound on Approximation Error, i.e., the loss in the approximation due to output dimension masking is at least:
\[
E_{\text{out}}^2 \geq f \sum_{i=1}^{r} \sigma_i^2 + \sum_{i = r+1}^{p} \sigma_i^2.
\]

The first term $f \sum_{i=1}^{r} \sigma_i^2$ represents the loss from masking out a fraction $f$ of the energy from the top $r$ singular values. The second term $\sum_{i = r+1}^{p} \sigma_i^2$ accounts for the singular values beyond rank $r$.

Now, Relating $s$ and $f$, from the parameter equality:
\[
s = \frac{r (d_s + d)}{2d} = \frac{r}{2} \left( 1 + \frac{d_s}{d} \right).
\]
Since $d_s = d (1 - f)$:

\begin{align*}
s &= \frac{r}{2} \left( 1 + \frac{d_s}{d} \right) \\
  &= \frac{r}{2} \left( 1 + (1 - f) \right) \\
  &= \frac{r}{2} (2 - f) \\
  &= r \left( 1 - \frac{f}{2} \right).
\end{align*}

Thus,
\[
\frac{s}{r} = 1 - \frac{f}{2}.
\]

We observe, in Rank Dimension Masking :
The approximation error comes from the discarded smaller singular values (indices $i > s$).
Since $s = r \left( 1 - \frac{f}{2} \right)$, we discard the smallest $r - s = r \left( \frac{f}{2} \right)$ singular values among the top $r$.
\[
E_{\text{rank}}^2 = \sum_{i=s+1}^{p} \sigma_i^2 = \sum_{i=r\left(1 - \frac{f}{2}\right)+1}^{p} \sigma_i^2.
\]

In Output Dimension Masking:
The error includes a loss from the largest singular values, scaled by $f$, because masking affects all components equally.
\[
E_{\text{out}}^2 \geq f \sum_{i=1}^{r} \sigma_i^2 + \sum_{i=r+1}^{p} \sigma_i^2.
\]

Total Energy of Top $r$ Singular Values
\[
q = \sum_{i=1}^{r} \sigma_i^2.
\]

Sum of Discarded Singular Values in Rank Masking:
\[
q' = \sum_{i = s+1}^{r} \sigma_i^2 = q - \sum_{i=1}^{s} \sigma_i^2.
\]

Relation between $f$ and $\frac{s}{r}$:
\[
f = 2 \left( 1 - \frac{s}{r} \right).
\]

Now, Expressing $E_{\text{rank}}^2$ and $E_{\text{out}}^2$ in terms of $q$ and $q'$, we get the Rank Dimension Masking Error:
\[
E_{\text{rank}}^2 = q' + \sum_{i = r+1}^{p} \sigma_i^2.
\]

Output Dimension Masking Error Lower Bound:
\[
E_{\text{out}}^2 \geq f q + \sum_{i = r+1}^{p} \sigma_i^2.
\]

Since $q' = q - \sum_{i=1}^{s} \sigma_i^2$ and $s = r \left( 1 - \frac{f}{2} \right)$, we have:

\begin{align*}
q' &= q - \sum_{i=1}^{s} \sigma_i^2 \\
   &\leq q - s \left( \frac{\sigma_r^2}{r} \cdot r \right) \quad (\text{since } \sigma_i \geq \sigma_r) \\
   &= q - s \sigma_r^2 \\
   &= q - r \left( 1 - \frac{f}{2} \right) \sigma_r^2.
\end{align*}

Therefore,
\[
E_{\text{rank}}^2 \leq q - r \left( 1 - \frac{f}{2} \right) \sigma_r^2 + \sum_{i = r+1}^{p} \sigma_i^2.
\]

Comparing with $E_{\text{out}}^2$, we get,
\[
E_{\text{out}}^2 \geq f q + \sum_{i = r+1}^{p} \sigma_i^2.
\]

Since $f = 2 \left( 1 - \frac{s}{r} \right)$, we have:

\[
E_{\text{out}}^2 \geq 2 \left( 1 - \frac{s}{r} \right) q + \sum_{i = r+1}^{p} \sigma_i^2.
\]

Therefore, the difference between $E_{\text{out}}^2$ and $E_{\text{rank}}^2$ is:


\begin{align*}
E_{\text{out}}^2 - E_{\text{rank}}^2 & \geq \left[ 2 \left( 1 - \frac{s}{r} \right) q - \left( q - r \left( 1 - \frac{f}{2} \right) \sigma_r^2 \right) \right] \\
   &= \left( 1 - \frac{s}{r} \right) \left( q + r \sigma_r^2 \right).
\end{align*}


Since $q \geq r \sigma_r^2$, the difference is non-negative, implying:

\[
E_{\text{rank}}^2 \leq E_{\text{out}}^2.
\]

\end{proof}









\begin{lemma}
\label{Lemma : lemma_1 }
Let \( m_c \in \mathbb{R}^{m \times n} \) be a matrix representing the content merger and \( m_s \in \mathbb{R}^{m \times n} \) be a matrix representing the style merger. The problem of minimizing the \( L_1 \)-norm of \( m_c \) subject to a rank constraint on \( m_c \) can be written as:
\[
\min \|m_c\|_1 \quad \text{subject to} \quad \text{rank}(m_c) > \text{rank}(m_s)
\]
This problem is non-convex due to the rank constraint. A convex relaxation can be achieved by approximating the rank of a matrix using the nuclear norm \( \| \cdot \|_* \), which is the sum of the singular values of the matrix. Thus, the original problem can be relaxed to:
\[
\min \|m_c\|_1 \quad \text{subject to} \quad \|m_c\|_* > \|m_s\|_*
\]
where \( \|m_c\|_* \) denotes the nuclear norm of \( m_c \), and \( \|m_s\|_* \) is the nuclear norm of \( m_s \).

This relaxed problem can be approached via a Lagrangian penalty formulation:
\[
\mathcal{L}(m_c, m_s, \lambda) = \|m_c\|_1 + \lambda \max(0, \|m_s\|_* - \|m_c\|_*)
\]
for some penalty parameter \( \lambda \geq 0 \), which enforces the constraint \( \|m_c\|_* > \|m_s\|_* \) in the limit as \( \lambda \to \infty \).
\end{lemma}

\begin{proof}
The rank of a matrix \( m_c \) is a non-convex function, making it difficult to optimize directly. The nuclear norm \( \|m_c\|_* \), defined as the sum of the singular values of \( m_c \), provides a \textit{convex envelope} of the rank function over the unit ball of matrices in the operator norm. Minimizing the nuclear norm encourages low-rank solutions because the nuclear norm penalizes the magnitude of singular values, making it an effective surrogate for the rank.

Thus, we replace the rank constraint \( \text{rank}(m_c) > \text{rank}(m_s) \) with the nuclear norm constraint:

\[
\|m_c\|_* > \|m_s\|_*
\]
This converts the non-convex constraint into a convex inequality that we can handle more easily in optimization.

Since the difference \( \|m_s\|_* - \|m_c\|_* \) is not convex, directly enforcing \( \|m_c\|_* > \|m_s\|_* \) would introduce non-convexity back into the problem. Instead, we approach this with a \textit{penalty function} that gradually enforces the constraint.

Define the Lagrangian-like penalty function:
\[
\mathcal{L}(m_c, m_s, \lambda) = \|m_c\|_1 + \lambda \max(0, \|m_s\|_* - \|m_c\|_*)
\]

where:

- \( \lambda \geq 0 \) controls the strength of the constraint enforcement.

- The penalty term \( \max(0, \|m_s\|_* - \|m_c\|_*) \) becomes zero if \( \|m_c\|_* \geq \|m_s\|_* \) and adds a positive penalty otherwise.

This penalty formulation turns the original constrained problem into an unconstrained optimization problem, where the constraint \( \|m_c\|_* > \|m_s\|_* \) is gradually enforced by increasing \( \lambda \).

\textbf{Convergence and Feasibility.}
As \( \lambda \to \infty \), the penalty for violating \( \|m_c\|_* > \|m_s\|_* \) becomes very large, making it infeasible for any solution to have \( \|m_c\|_* \leq \|m_s\|_* \) in the limit. Therefore, the solution to the penalized Lagrangian problem approaches the solution to the original problem:
\[
\min \|m_c\|_1 \quad \text{subject to} \quad \|m_c\|_* > \|m_s\|_*
\]

In other words, by iteratively solving for \( m_c \) with larger values of \( \lambda \), we approximate a solution that satisfies the nuclear norm constraint. This penalty approach provides a feasible, convex approximation for the non-convex problem, yielding a solution that respects the desired rank constraint indirectly.

The penalty formulation of the Lagrangian:
\[
\mathcal{L}(m_c, m_s, \lambda) = \|m_c\|_1 + \lambda \max(0, \|m_s\|_* - \|m_c\|_*)
\]
provides an effective convex relaxation for the non-convex constraint \( \text{rank}(m_c) > \text{rank}(m_s) \). This approach ensures that the solution minimizes \( \|m_c\|_1 \) while approximately satisfying the rank constraint in a manner that is computationally feasible and does not require convexity of the difference \( \|m_s\|_* - \|m_c\|_* \).

\end{proof}

\begin{algorithm}
\caption{Merging LoRA with Cycle-Consistency Loss}
\SetKwInOut{Input}{Input}
\SetKwInOut{Output}{Output}
\Input{Content images $I_c$, Style images $I_s$}
\Output{Merged LoRA $L_m$}
\BlankLine
\textbf{Step 1: Train LoRA for content and style}\\
- Learn content LoRA ($L_c$) using content images $I_c$ and prompt \texttt{$p_c$=``a <V1> object in <S1> style''}\\
- Learn style LoRA ($L_s$) using style images $I_s$ and prompt \texttt{$p_s$=``a <V2> object in <S2> style''}
\BlankLine
\textbf{Step 2: Apply cycle-consistency loss across style} \\
\tcp*[f]{Generate variations of $I_c$} \\
- $I_{cc} \gets (D + L_c)(I_c, \text{``<V1> object in <S1> style''})$  \\
\tcp*[f]{Add style} \\
- $I_{cs} \gets (D + L_s)(I_{c}, \text{``<V1> object in <S2> style''})$ \\
\tcp*[f]{Remove style} \\
- $I_{csc} \gets (D + L_c)(I_{cs}, \text{``<V1> object in <S1> style''})$ \\
\tcp*[f]{Ensure cycle-consistency loss across style}\\
- $\mathcal{L_{\text{cycle\_sty}}} \gets \text{MSE}(I_{cc}, I_{csc})$ \\
\BlankLine
\textbf{Step 3: Apply cycle-consistency loss across object} \\
\tcp*[f]{Generate variations of $I_s$} \\
- $I_{ss} \gets (D + L_s)(I_s, \text{``<V2> object in <S2> style''})$ \\
\tcp*[f]{Add object} \\
- $I_{sc} \gets (D + L_c)(I_{s}, \text{``<V1> object in <S2> style''})$  \\
\tcp*[f]{Remove object} \\
- $I_{scs} \gets (D + L_s)(I_{sc}, \text{``<V2> object in <S2> style''})$ \\
\tcp*[f]{Ensure cycle-consistency loss across content}\\
- $\mathcal{L_{\text{cycle\_content}}} \gets \text{MSE}(I_{ss}, I_{scs})$ \\
\BlankLine
\textbf{Step 4: Merging LoRAs with consistency loss}\\
- Train merged LoRA $L_m$ using $L_c$ and $L_s$ with the consistency loss:
\begin{align*}
    \mathcal{L}_{constyle} &= ||(D + L_m)(I_c, p_c) - (D + L_c)(I_c, p_c)|| \\
    & + ||(D + L_m)(I_s, p_s) - (D + L_s)(I_s, p_s)|| \\
    &+ \lambda_{cycle} \cdot \mathcal{L}_{\text{cycle\_sty}} + \lambda_{cycle} \cdot \mathcal{L}_{\text{cycle\_content}}
\end{align*}
where $D$ is the T2I diffusion model and $\lambda_{cycle}$ is the scaling factor.
\BlankLine
\textbf{Step 5: Inference}\\
- During inference, pass the combined trained tokens as prompt (e.g., \texttt{``a <V1> object in <S2> style running''}) to the T2I diffusion model with merged LoRA $L_m$ to generate variations corresponding to the text prompt.
\label{alg:cycle_consistent_merging}
\end{algorithm}

\begin{algorithm}
\SetAlgoLined
\KwIn{Content merger $m_c$, Style merger $m_s$, content threshold ($T_{content}$), style threshold ($T_{style}$)}
\KwOut{Initialized content merger vector, Initialized style merger vector}

$V \gets \text{rand}(64, 1)$\;
$V' \gets \frac{V}{\|V\|}$ \tcp*{Normalize V}

\uIf{$Rank(m_c) > Rank(m_s)$}{
    $\text{content\_merger\_init} \gets 1(V'>T_{style})$\\
    $\text{style\_merger\_init} \gets 1(V'>T_{content})$
}
\uElseIf{$Rank(m_c) < Rank(m_s)$}{
    $\text{content\_merger\_init} \gets 1(V'>T_{content})$\\
    $\text{style\_merger\_init} \gets 1(V'>T_{style})$
}
\Else{
    $\text{content\_merger\_init} \gets \mathbf{1}_{64 \times 1}$\\
    $\text{style\_merger\_init} \gets \mathbf{1}_{64 \times 1}$
}

\caption{Content and Style Merger Initialization Algorithm}
\label{alg:initialization}
\end{algorithm}

\section{Algorithm details}

In this section, we provide additional details for our approach DuoLoRA. The algorithm for cycle-consistent merging using constyle loss has been provided in Algorithm.~\ref{alg:cycle_consistent_merging}. Also, the algorithm for content and style mergers initialization method is provided in Algorithm.~\ref{alg:initialization}.

\section{Details of user study}

Since the perceptual metrics are not always reliable, we conduct user study to verify the efficacy of our method.
We provide 20 examples of content, style pairs and corresponding generated images using Naive merging, B-LoRA, ZipLoRA and DuoLoRA. Then, we asked the following question to amazon mechanical turks: "which of the generated images is of best visual quality considering factors that we preserve both the content and style?", and the options are ``Naive merging'', ``B-LoRA'', ``ZipLoRA'', ``DuoLoRA'', ``None is satisfactory''. We evaluate this by 50 users, totalling 1000 questionnaires by Amazon Mechanical Turk (AMT) to get unbiased results.
The aggregate responses in Table.~\ref{table:user_study} showed that DuoLoRA generated images significantly outperformed the baselines by a large margin (50\%). This verify that DuoLoRA generated images retain both content and style.

\begin{table}[!h]
  \small
  \centering
  \vspace{-0.2cm}
  \caption{\small{User study}}
  \vspace{-0.2cm}
  \scalebox{0.95}{
  \begin{tabular}{lccccc}
  \toprule
  {\bf None} & {\bf Naive merging} & {\bf B-LoRA} & {\bf ZipLoRA} & {\bf DuoLoRA} \\ 
  \midrule
  {0.0\%} & {10\%} & {18\%} & {22\%} & {50\%} \\
  \bottomrule
  \end{tabular}}
  \vspace{-0.2cm}
  \label{table:user_study}
\end{table}

\section{Additional results and ablations}

\begin{table*}[h]
\centering
\caption{Performance comparison of content and style merging across different datasets and methods}
\vspace{-0.2cm}
\scalebox{0.85}{
\begin{tabular}{lccc|ccc|ccc|c}
\hline
\textbf{Method} & \multicolumn{3}{c|}{\textbf{Dreambooth + StyleDrop}} & \multicolumn{3}{c|}{\textbf{Subjectplop}} & \multicolumn{3}{c|}{\textbf{Subjectplop + StyleDrop}} & \textbf{\# Params} \\
 & \textbf{DINO} & \textbf{CLIP-I} & \textbf{CSD-s} & \textbf{DINO} & \textbf{CLIP-I} & \textbf{CSD-s} & \textbf{DINO} & \textbf{CLIP-I} & \textbf{CSD-s} & \\
\hline
Naïve Merging                 & 0.47 & 0.64 & 0.44 & 0.48 & 0.59 & 0.30 & 0.42 & 0.49 & 0.12 & -     \\
Joint Training                & 0.55 & 0.58 & 0.22 & 0.63    & 0.54    & 0.35    & 0.69    & 0.56    & 0.15    & 2.6B     \\
B-LoRA~\cite{frenkel2024implicit} (ECCV'24) & 0.45 & 0.57 & 0.28 & 0.64 & 0.57 & 0.32 & 0.63 & 0.56 & 0.14 & -     \\
ZipLoRA~\cite{shah2025ziplora} (ECCV'24)    & 0.53 & 0.55 & 0.41 & 0.75 & 0.62 & 0.35 & 0.87 & 0.56 & 0.16 & 1.33M \\
\rowcolor{gray!20} ZipRank                  & 0.53 & 0.64 & 0.42 & 0.71 & 0.62 & 0.35 & 0.86 & 0.56 & 0.17 & 0.07M \\
\rowcolor{gray!20} ZipRank + Layer-Priors & 0.54 & 0.67 & 0.45 & 0.73 & 0.63 & 0.37 & 0.90 & 0.56 & 0.18 & 0.07M \\
\rowcolor{gray!20} DuoLoRA & \textbf{0.56} & \textbf{0.69} & \textbf{0.48} & \textbf{0.78} & \textbf{0.65} & \textbf{0.40} & \textbf{0.90} & \textbf{0.58} & \textbf{0.20} & \textbf{0.07M} \\
\hline
\end{tabular}}
\label{tab:baseline_comparison}
\end{table*}

\subsection{Comparisons and Qualitative results}
We provide additional results using joint training baselines for all the datasets in Table.~\ref{tab:baseline_comparison}. Qualitative results in Fig.~\ref{fig:supple_qual_1}, Fig.~\ref{fig:supple_qual_2}, Fig.~\ref{fig:supple_qual_3}, Fig.~\ref{fig:supple_qual_dog2} and Fig.~\ref{fig:supple_qual_dog6} also shows that DuoLoRA performs better than the baselines. In Fig.~\ref{fig:qual_2}, we show ablations of how each components affect the merging. We also show results while using concepts from real-world concept-centric custom101 dataset~\cite{kumari2023multi} and styles from styledrop dataset. DuoLoRA outperforms baselines as shown in Fig.~\ref{fig:custom101_1}, Fig.~\ref{fig:custom101_2}, Fig.~\ref{fig:custom101_2}, Fig.~\ref{fig:custom101_3} and Fig.~\ref{fig:custom101_4}.
We also compare with Paircustomization~\cite{jones2024customizing} using their setup, since they require 1-shot concept-style pair for training. We use their dataset for fair comparison with 6 objects, 2 styles. DuoLoRA performs better than Paircustomization as shown in Fig.~\ref{fig:paircustom_compare} and also in main paper. Moreover, training DuoLoRA is more easier and parameter efficient than Paircustomization, which requires sort of joint training framework.

\subsection{Multi-concept stylization}
Here we provide details of multi-concept stylization.
We further extend our approach to handle multi-concept stylization. Given two concepts, \( C_1 \) and \( C_2 \), and a style \( S \), our objective is to generate an image that contains both \( C_1 \) and \( C_2 \) in style \( S \). To achieve this, we decompose the task into individual content-style merging processes, specifically \( C_1 \)-\( S \) and \( C_2 \)-\( S \) merging, using layer-prior-informed loss. We then perform an arithmetic merging of the outputs from \( C_1 \)-\( S \) and \( C_2 \)-\( S \) to create the final image. The steps are as follows:

\begin{itemize}
    \item We begin by merging each concept \( C_1 \) and \( C_2 \) with the style \( S \). First, we train LoRAs \( L_1 \), \( L_2 \), and \( L_s \) with identifiers \( <v1> \), \( <v2> \), and \( <s> \), respectively.
    \item Next, we merge LoRAs \( L_1 \) and \( L_2 \) with \( L_s \) in the rank dimension, applying the layer-prior loss as previously described. That is, we define the merged LoRAs as \( L_{1m} = \text{merge}(L_1, L_s) \) and \( L_{2m} = \text{merge}(L_2, L_s) \).
    \item After generating the merged LoRAs, we perform an arithmetic merging to obtain the final merged LoRA: \( L_{1,2,S} = \alpha_{1} L_{1m} + \alpha_{2} L_{2m} \).
    \item During inference, we use directional prompting with the merged LoRA \( L_{1,2,S} \), using \texttt{p = ``a $<v1>$ object on the left and a $<v2>$ object on the right in $<S>$ style''}. We find that directional prompting plays a crucial role in achieving high fidelity when generating multiple objects.
\end{itemize}

We extend this for 2, 3, 4 concepts from Dreambooth dataset and syles from StyleDrop dataset. The results are shown in Fig.~\ref{fig:2_concept} and Fig.~\ref{fig:3_concept}. We also show ablation for directional prompting in Fig.~\ref{fig:directional prompt}, which remains important for multi-concept composition.

\subsection{Recontextualization}

We also evaluate the recontextualization ability of our method. We use text prompts \texttt{'riding a boat', 'sleeping', 'riding a bicycle', 'riding a car', 'wearing a hat'} to generate different variation of styled concepts as shown in Fig.~\ref{fig:recontext_1} and Fig.~\ref{fig:recontext_2}. Our method can successfully recontextualize w.r.t the text prompts. 


\begin{figure*}
    \centering
        \centering
        \includegraphics[width=0.95\textwidth]{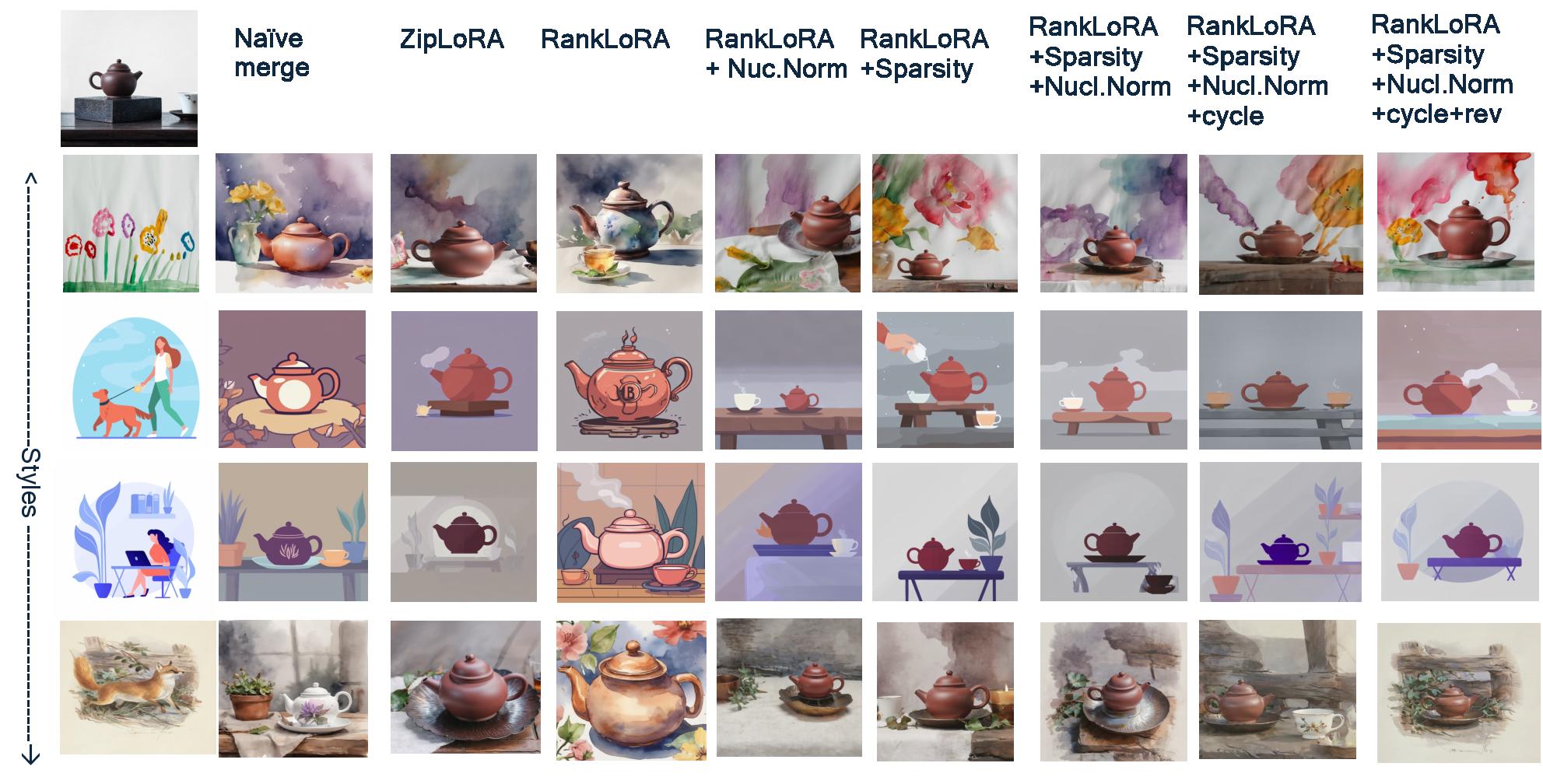} 
        \caption{Qualitative Results showing how each components impacting the merging}
        \label{fig:qual_2}   
\end{figure*}

\begin{figure*}
    \centering
        \centering
        \includegraphics[width=0.95\textwidth]{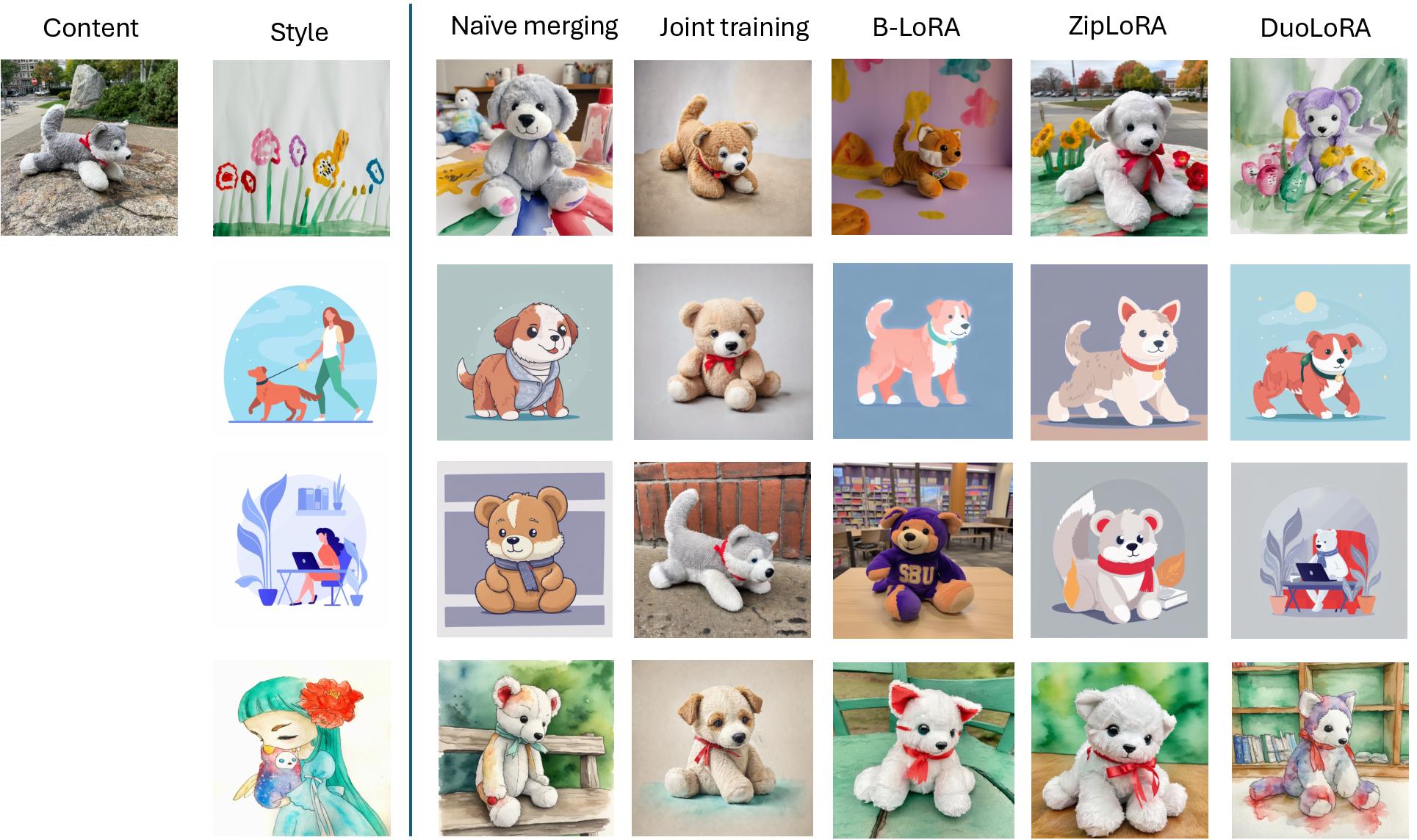} 
        \caption{Qualitative Results on Dreambooth + StyleDrop}
        \label{fig:supple_qual_1}   
\end{figure*}

\begin{figure*}
    \centering
        \centering
        \includegraphics[width=0.95\textwidth]{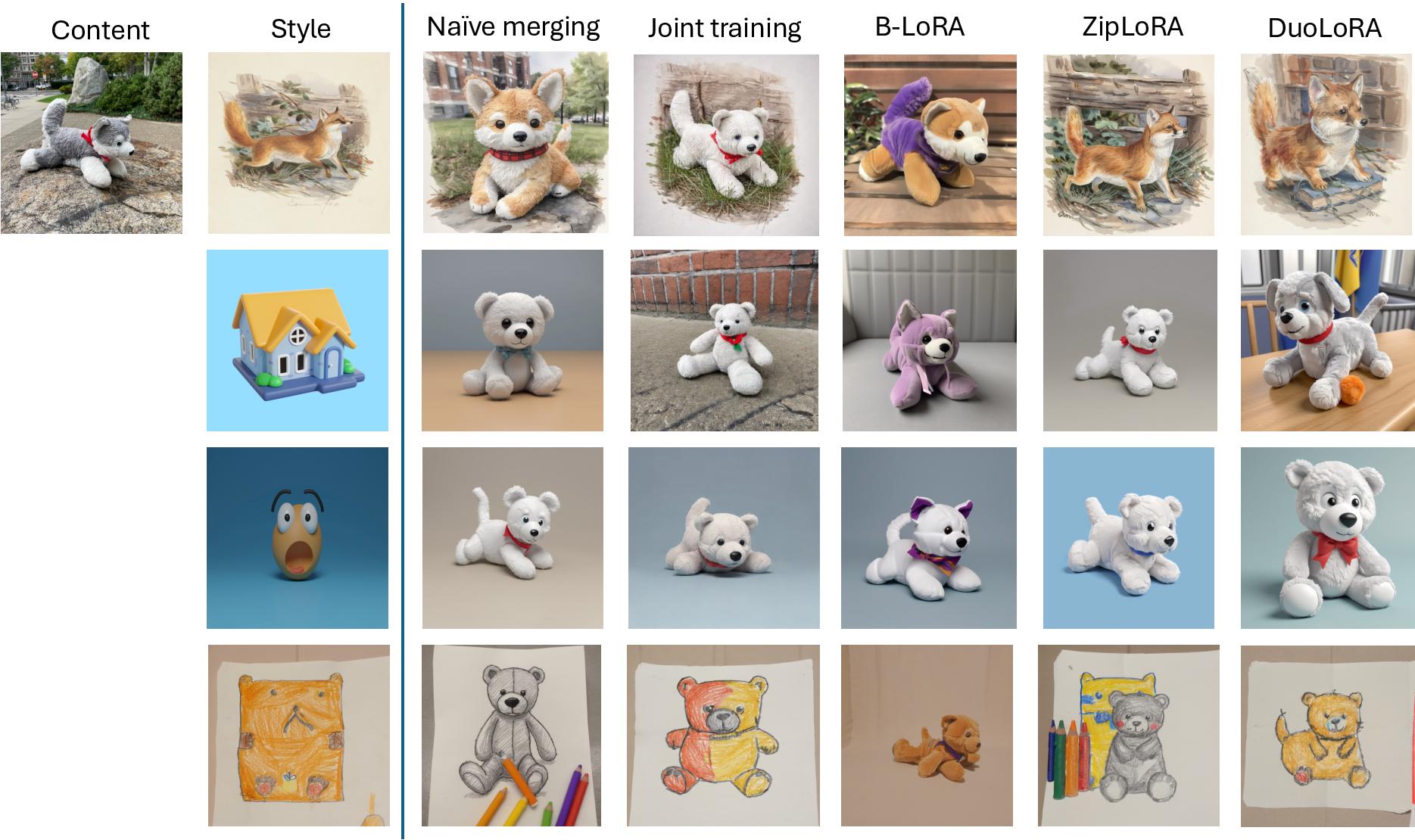} 
        \caption{Qualitative Results on Dreambooth + StyleDrop}
        \label{fig:supple_qual_2}   
\end{figure*}

\begin{figure*}
    \centering
        \centering
        \includegraphics[width=0.95\textwidth]{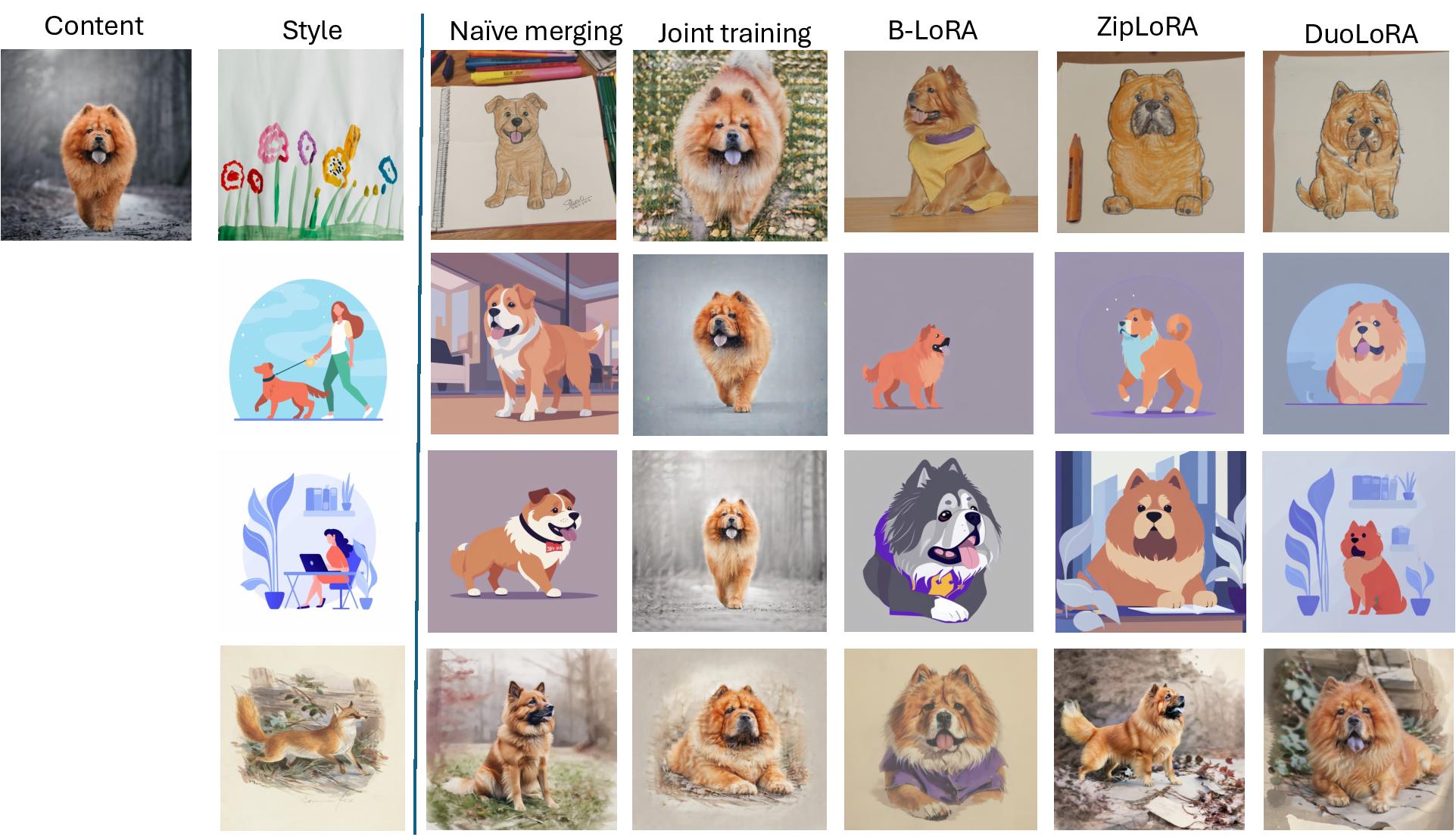} 
        \caption{Qualitative Results on Dreambooth + StyleDrop}
        \label{fig:supple_qual_dog2}   
\end{figure*}

\begin{figure*}
    \centering
        \centering
        \includegraphics[width=0.95\textwidth]{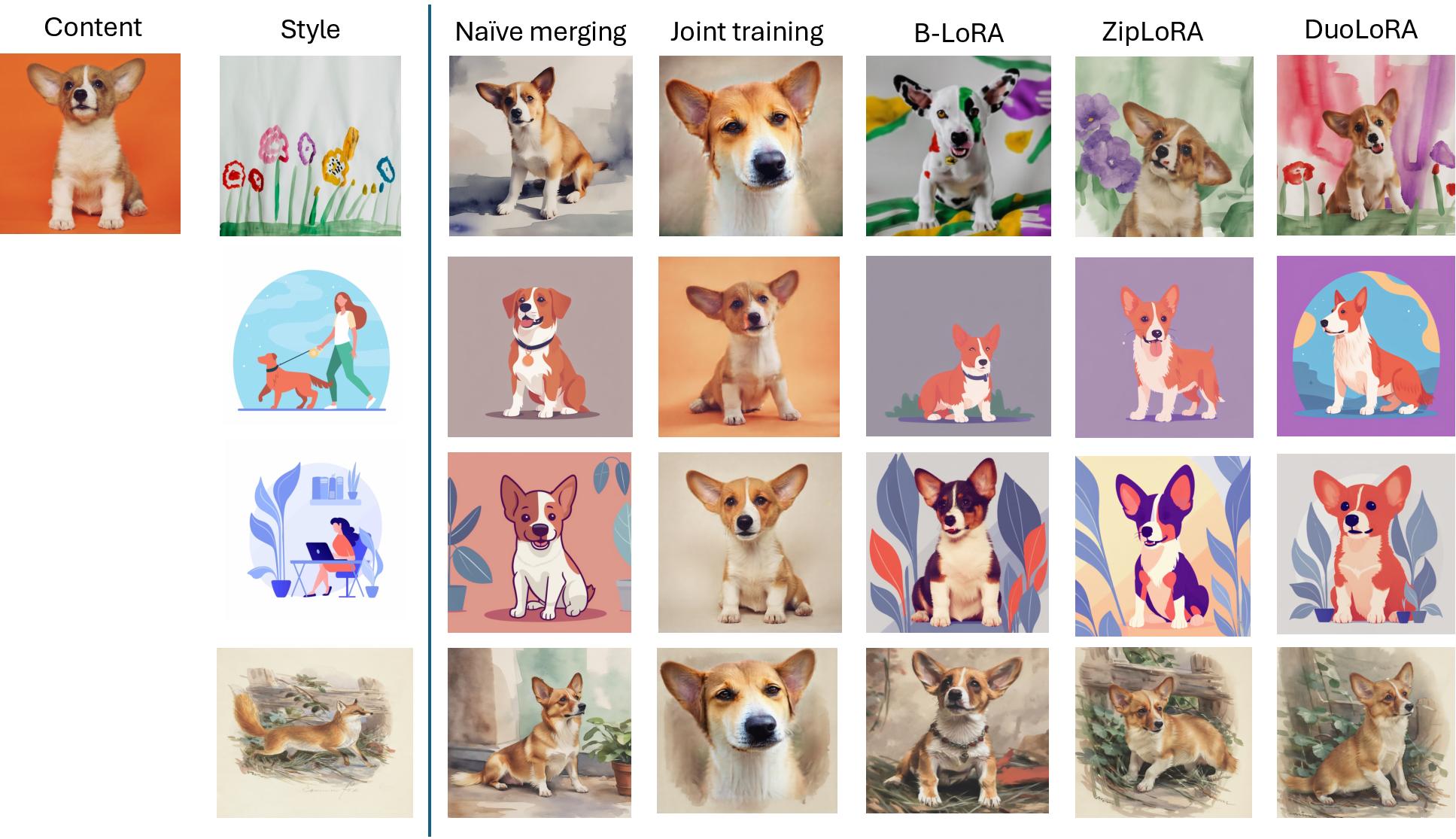} 
        \caption{Qualitative Results on Dreambooth + StyleDrop}
        \label{fig:supple_qual_dog6}   
\end{figure*}

\begin{figure*}
    \centering
        \centering
        \includegraphics[width=0.95\textwidth]{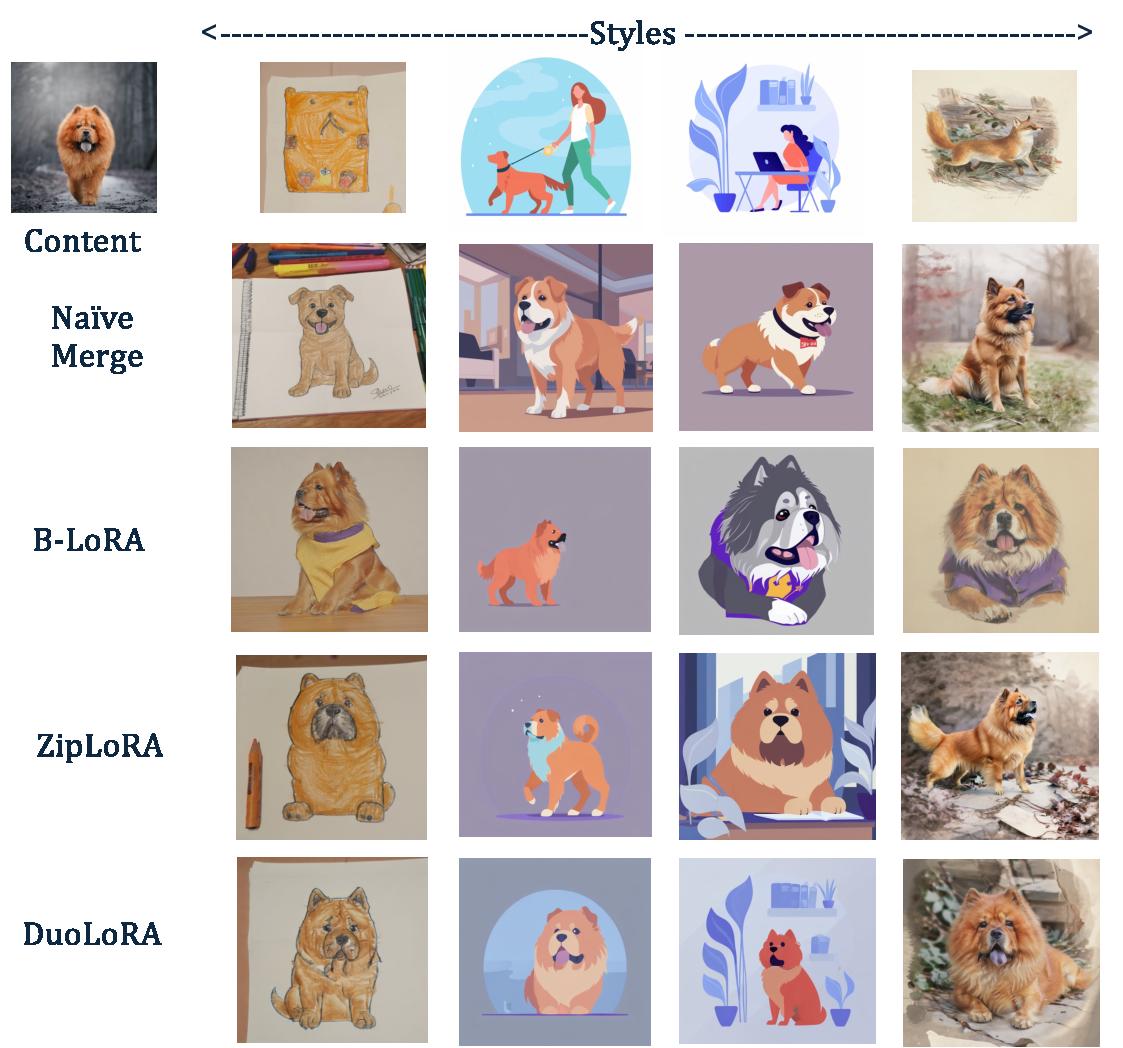}
        \caption{Qualitative Results on Dreambooth + StyleDrop.}
        \label{fig:qual_dog2}   
\end{figure*}

\begin{figure*}
    \centering
        \centering
        \includegraphics[width=1.0\textwidth]{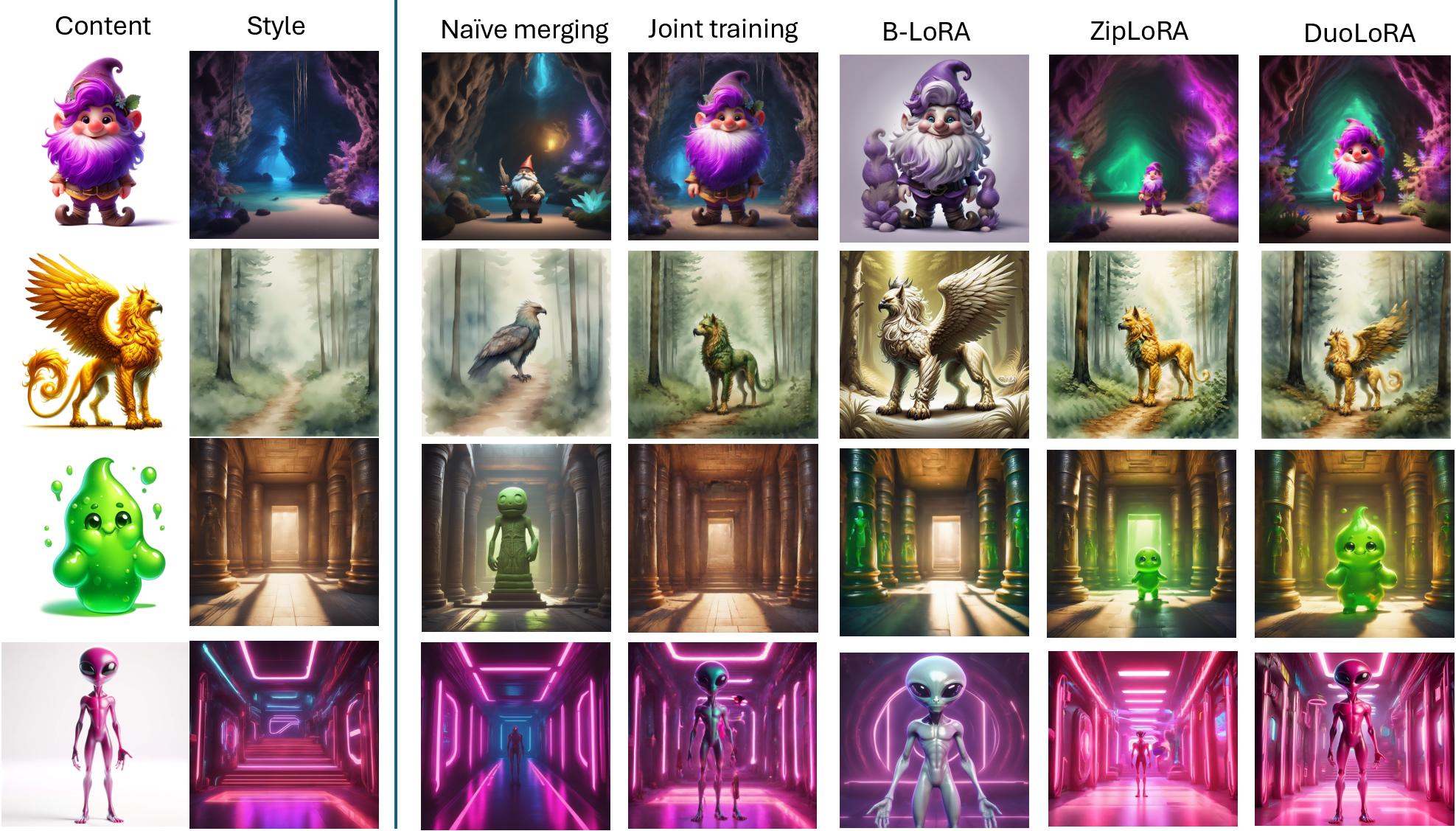} 
        \caption{Qualitative Results on Subjectplop}
        \label{fig:supple_qual_3}   
\end{figure*}

\begin{figure*}
    \centering
        \centering
        \includegraphics[width=0.95\textwidth]{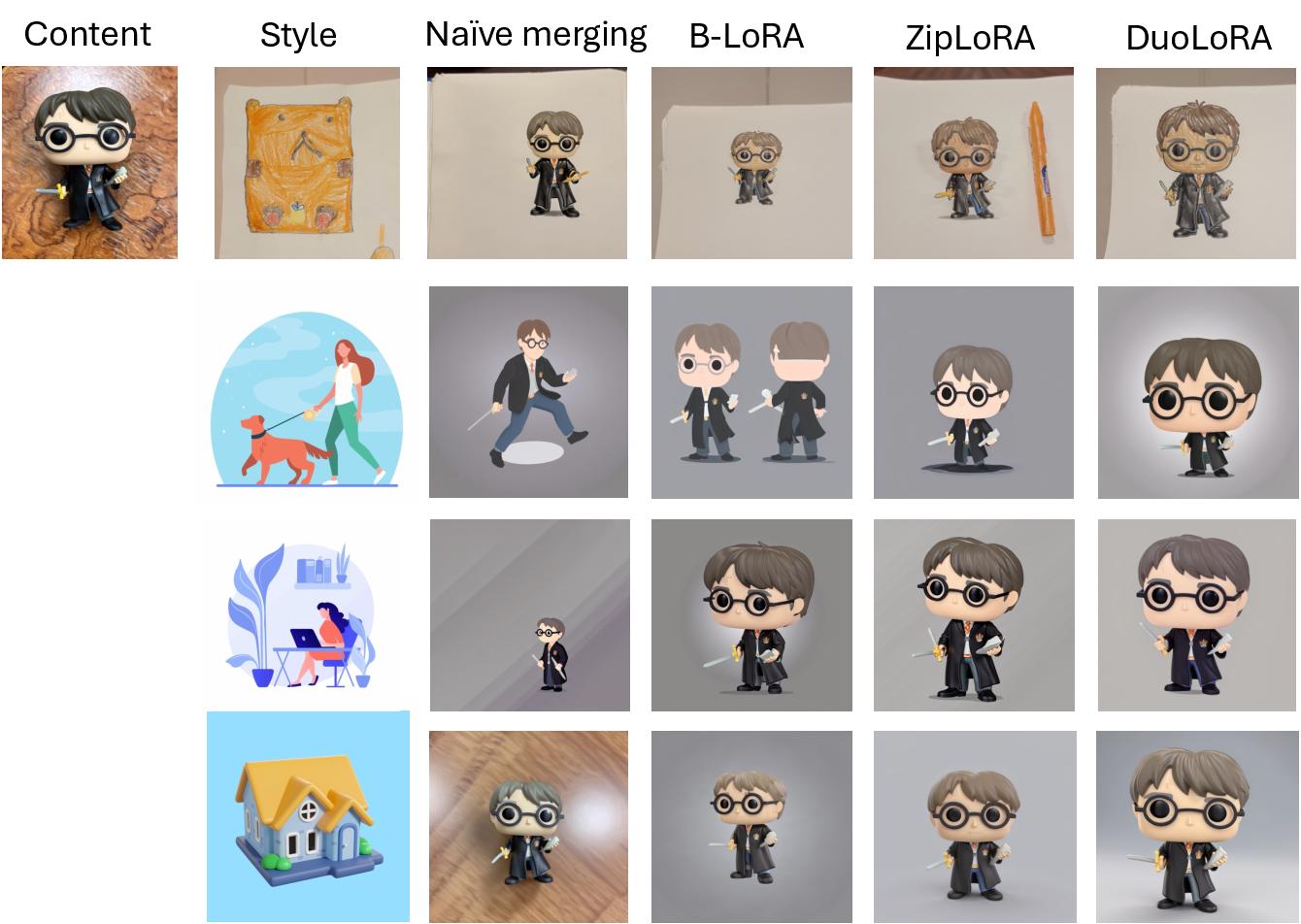} 
        \caption{Qualitative Results on Custom101}
        \label{fig:custom101_1}   
\end{figure*}

\begin{figure*}
    \centering
        \centering
        \includegraphics[width=0.95\textwidth]{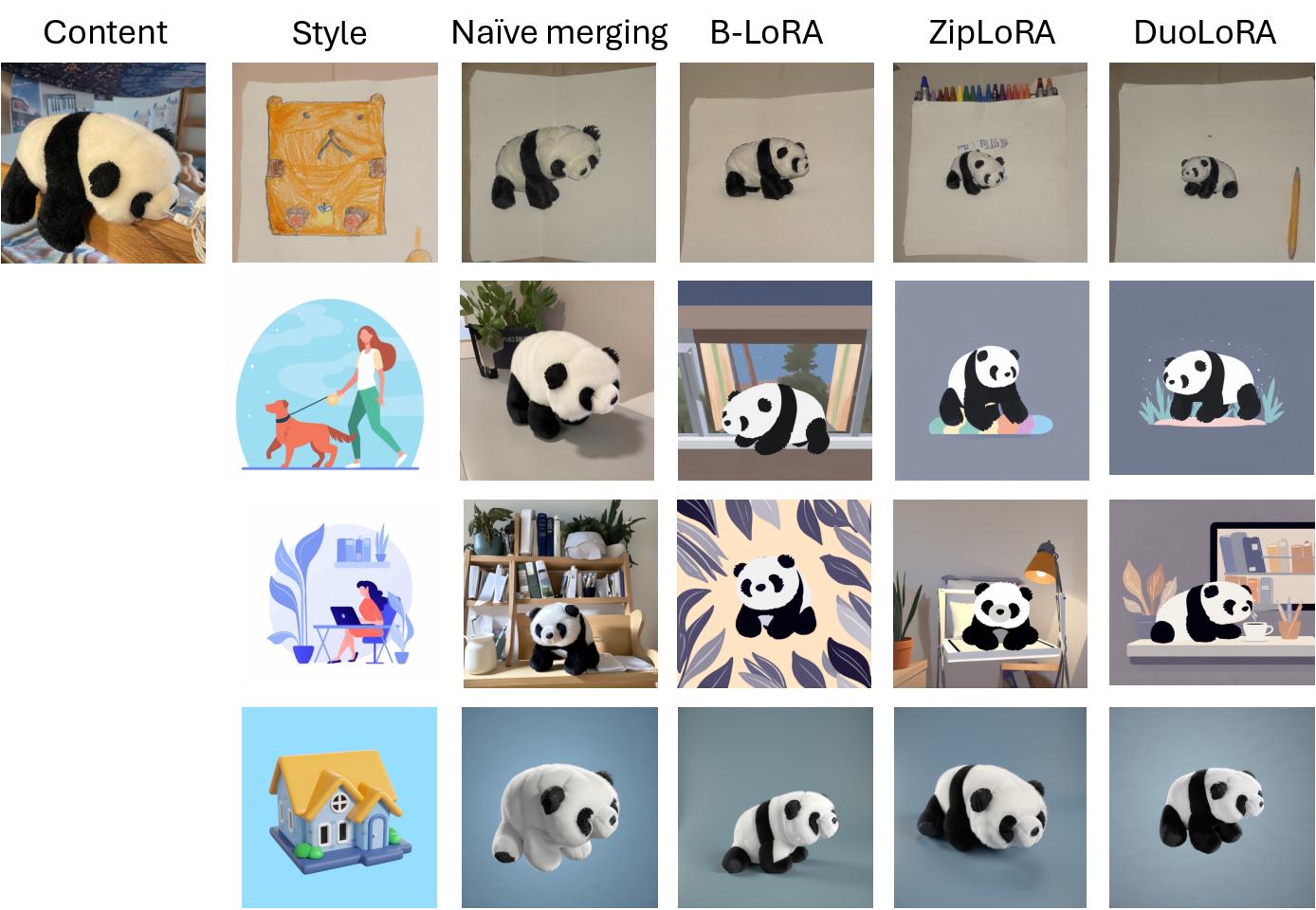} 
        \caption{Qualitative Results on Custom101}
        \label{fig:custom101_2}   
\end{figure*}

\begin{figure*}
    \centering
        \centering
        \includegraphics[width=0.95\textwidth]{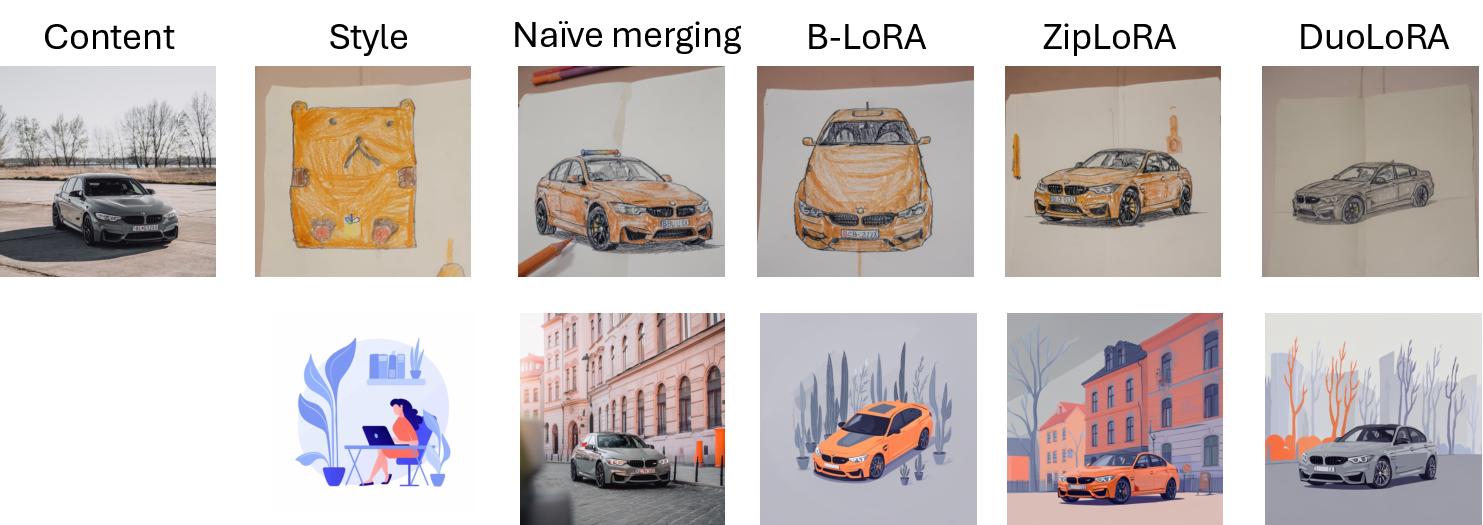} 
        \caption{Qualitative Results on Custom101}
        \label{fig:custom101_3}   
\end{figure*}

\begin{figure*}
    \centering
        \centering
        \includegraphics[width=0.95\textwidth]{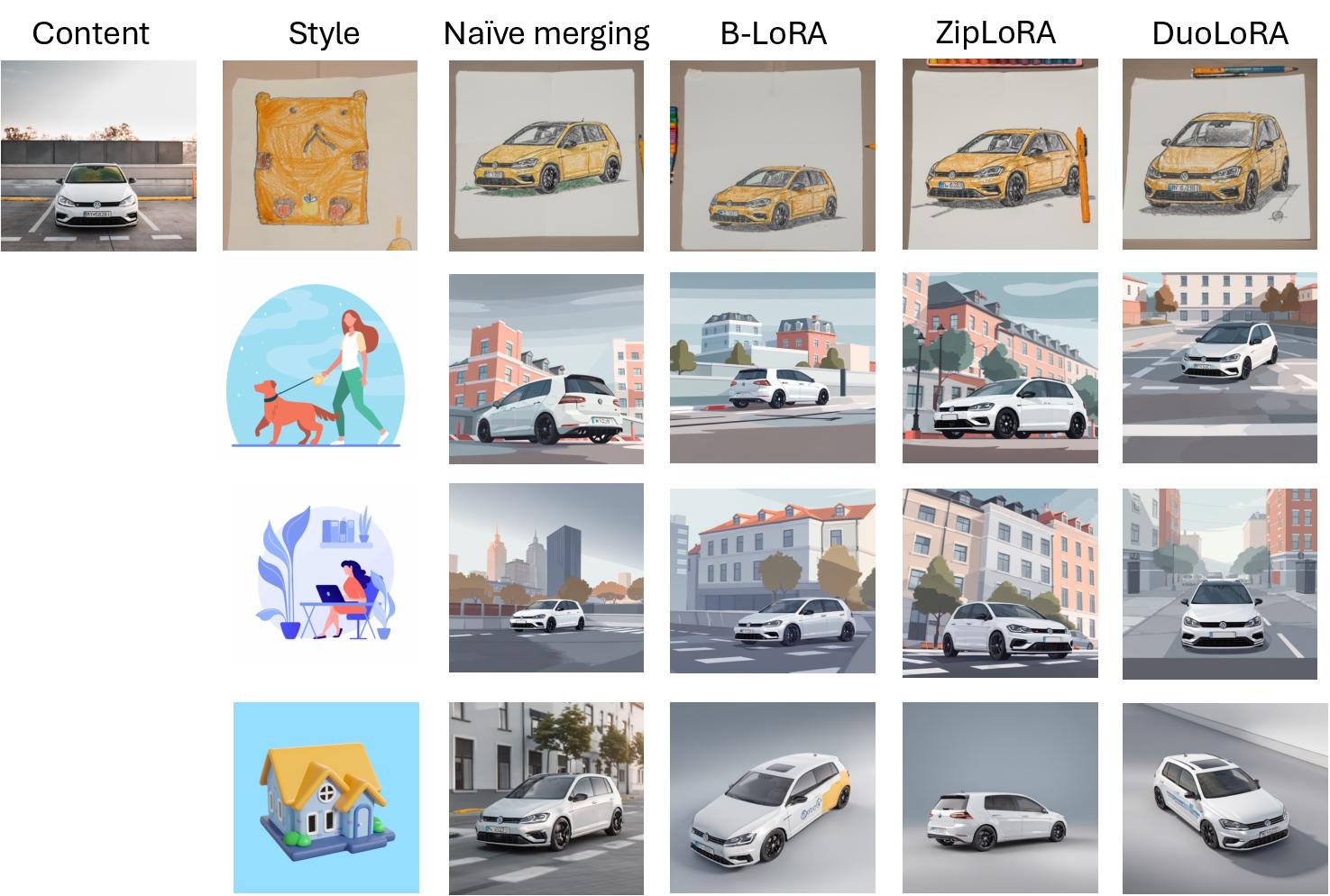} 
        \caption{Qualitative Results on Custom101}
        \label{fig:custom101_4}   
\end{figure*}

\begin{figure*}
    \centering
        \centering
        \includegraphics[width=0.95\textwidth]{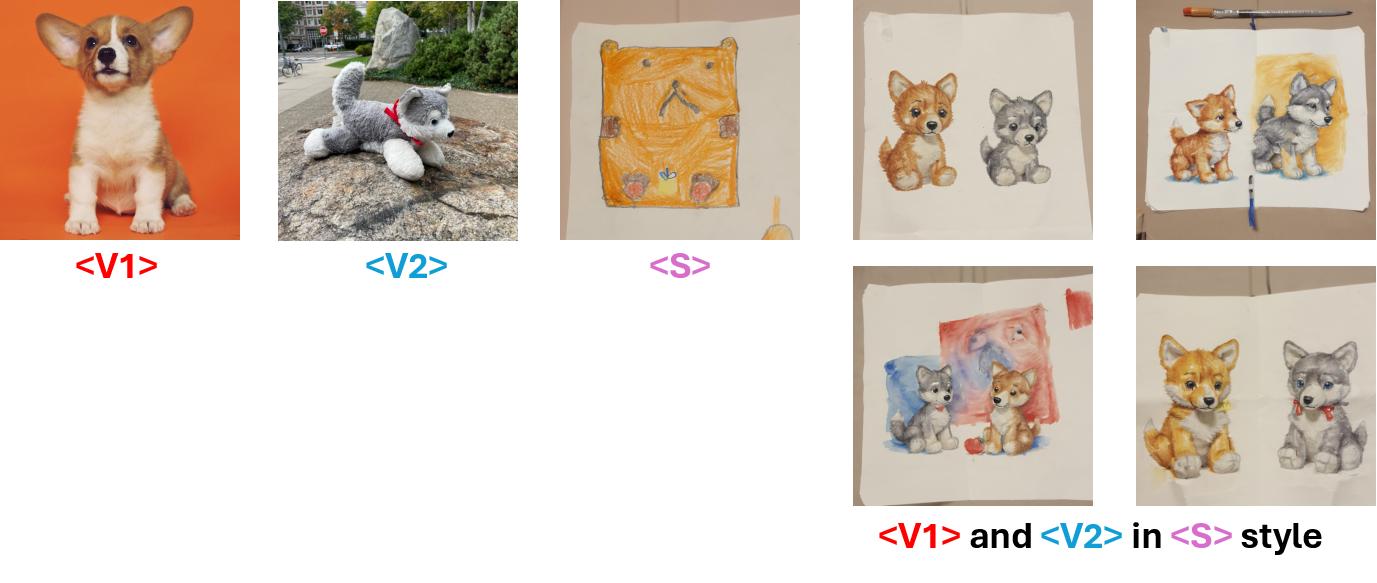} 
        \caption{Qualitative Results 2 concepts composition}
        \label{fig:2_concept}   
\end{figure*}

\begin{figure*}
    \centering
        \centering
        \includegraphics[width=0.95\textwidth]{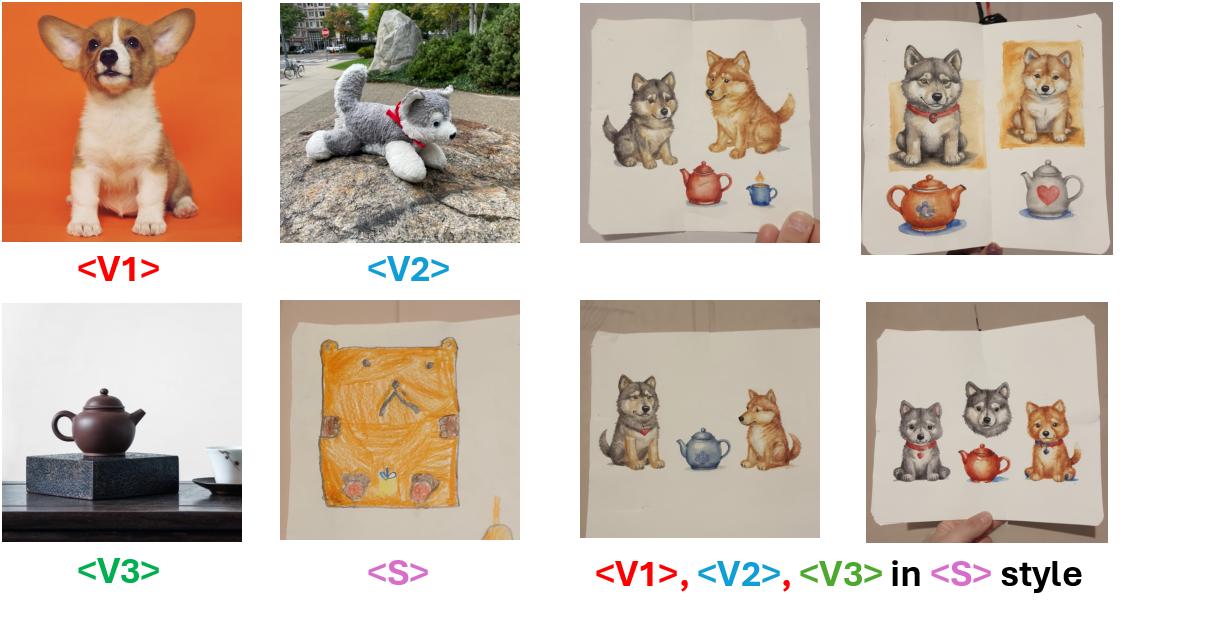} 
        \caption{Qualitative Results on 3 concepts composition}
        \label{fig:3_concept}   
\end{figure*}

\begin{figure*}
    \centering
        \centering
        \includegraphics[width=0.95\textwidth]{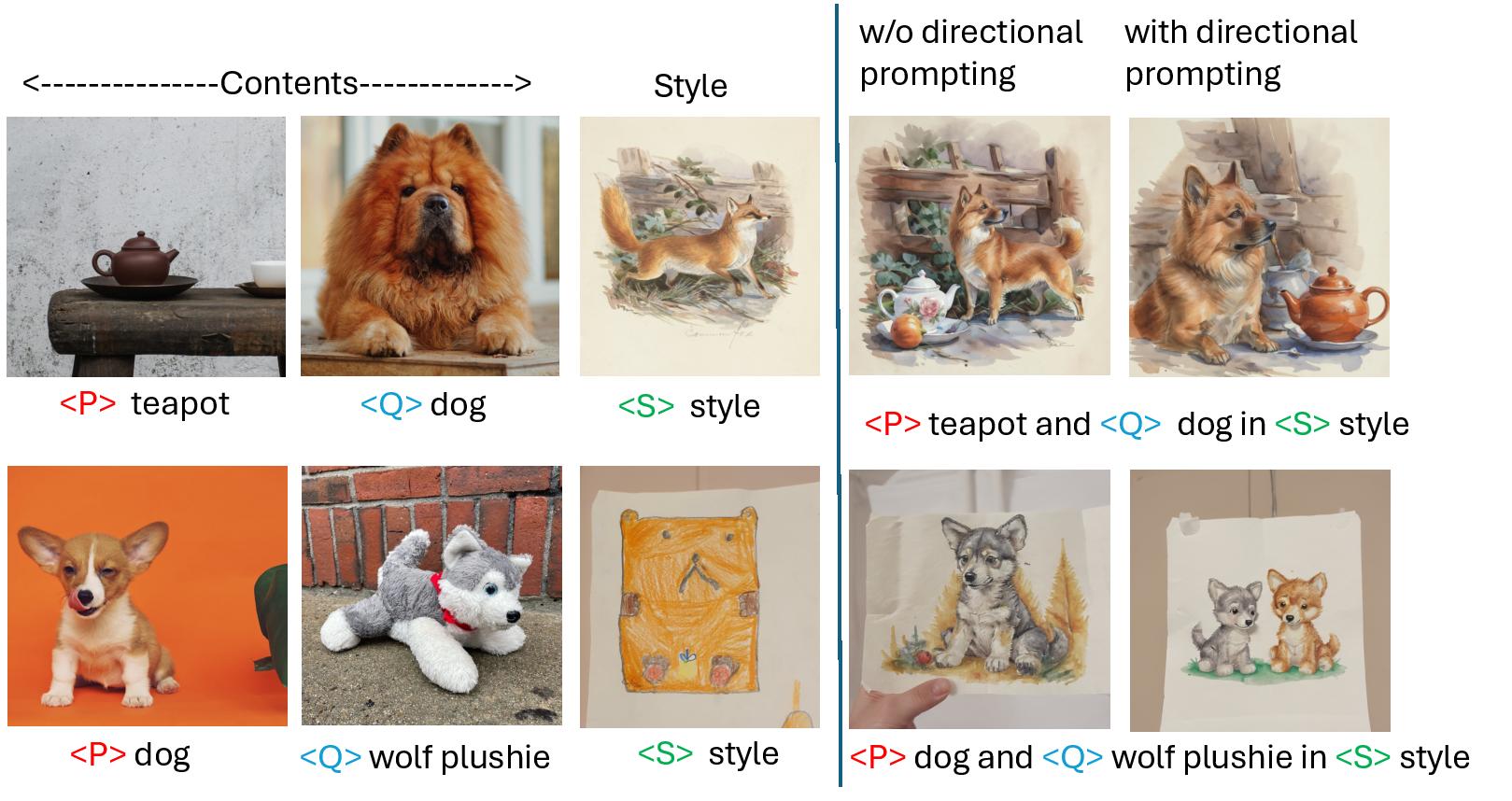} 
        \vspace{-0.2cm}
        \caption{Directional prompt ablation.}
        \vspace{-0.2cm}
        \label{fig:directional prompt}   
\end{figure*}

\begin{figure*}
    \centering
        \centering
        \includegraphics[width=1.0\textwidth]{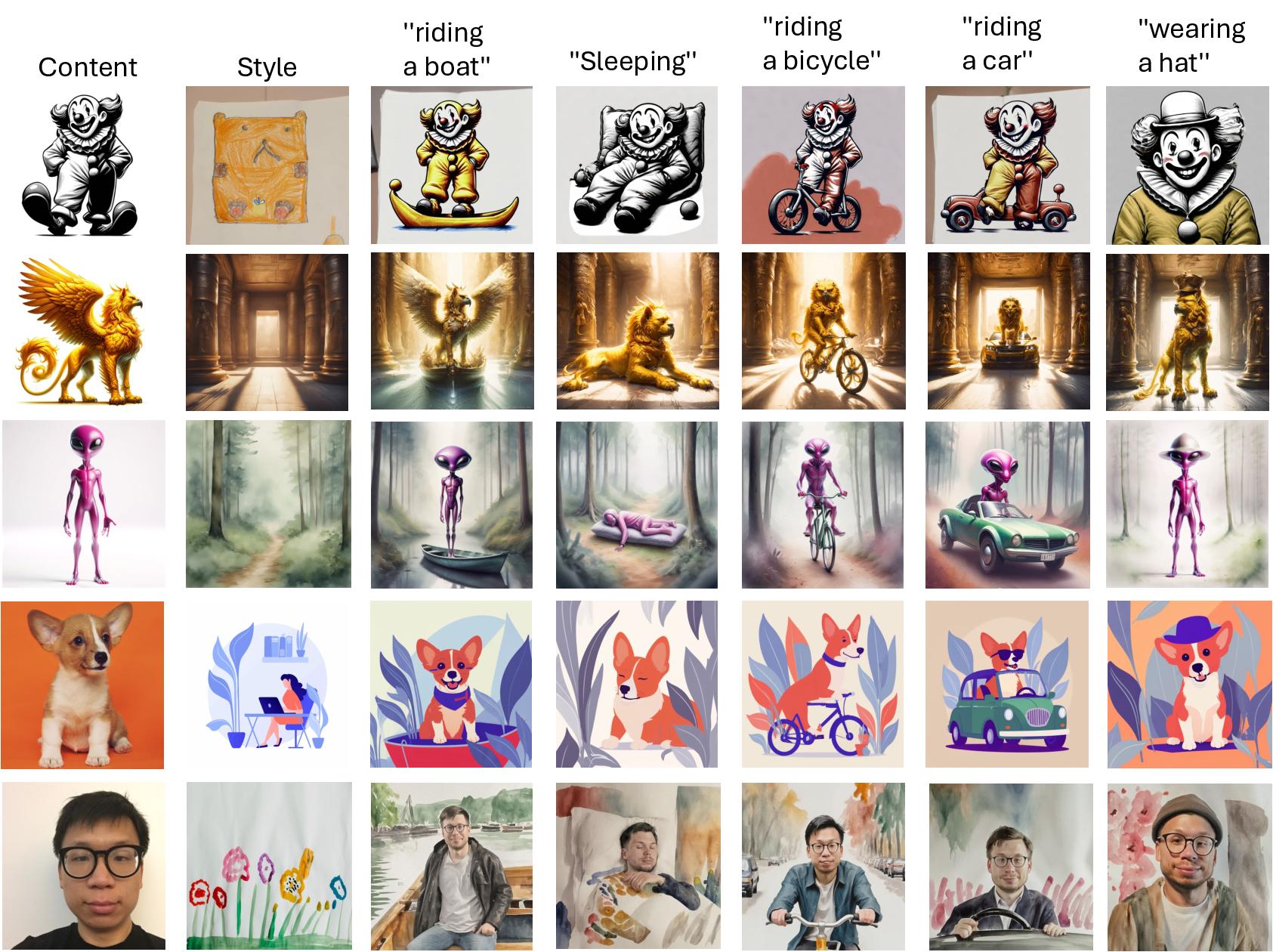} 
        \vspace{-0.2cm}
        \caption{Recontextualization examples}
        \vspace{-0.2cm}
        \label{fig:recontext_1}   
\end{figure*}

\begin{figure*}
    \centering
        \centering
        \includegraphics[width=1.0\textwidth]{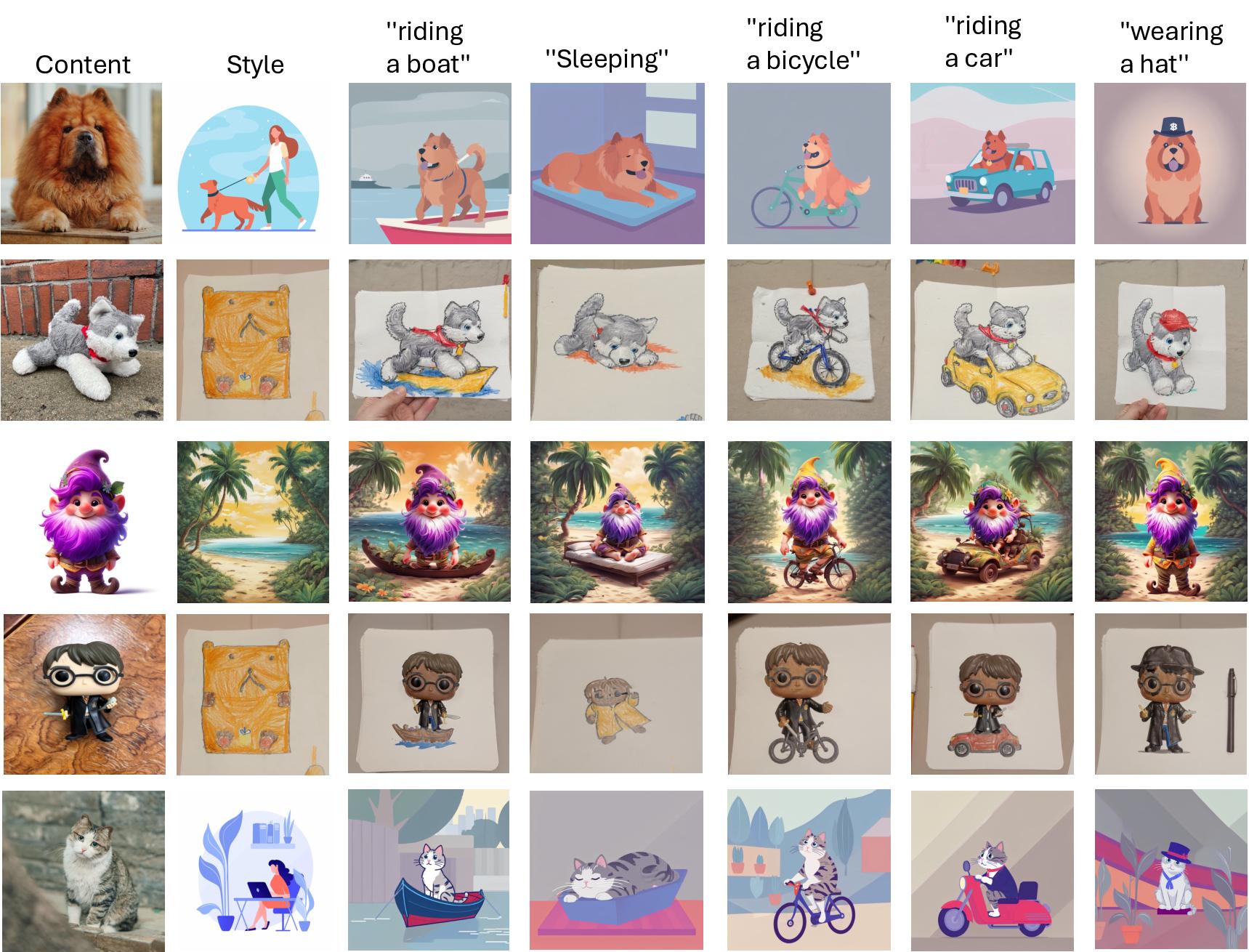} 
        \vspace{-0.2cm}
        \caption{Recontextualization examples}
        \vspace{-0.2cm}
        \label{fig:recontext_2}   
\end{figure*}

\begin{figure*}
    \centering
        \centering
        \includegraphics[width=0.75\textwidth]{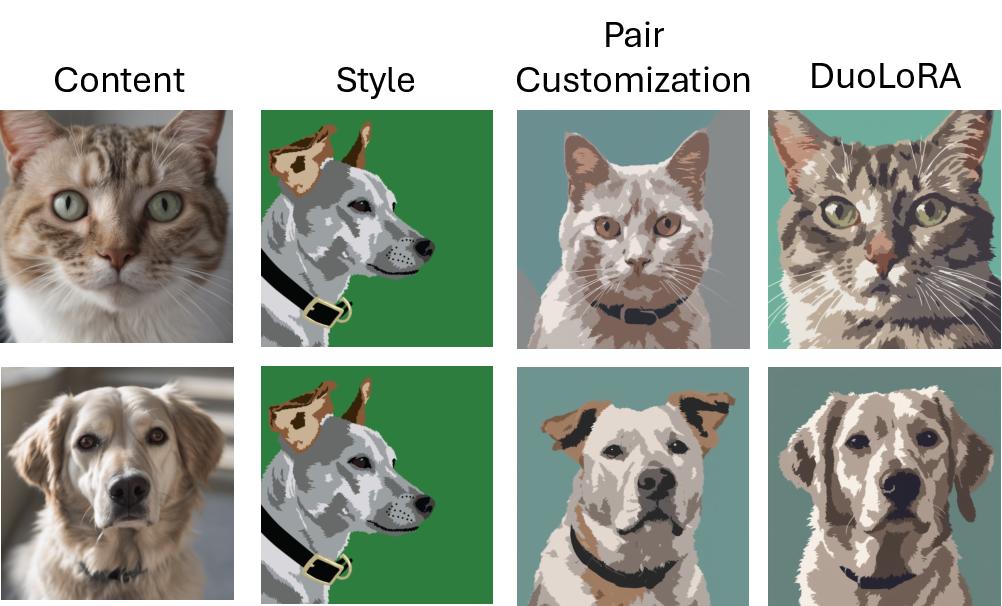} 
        \vspace{-0.2cm}
        \caption{Comparison with Paircustomization}
        \vspace{-0.2cm}
        \label{fig:paircustom_compare}   
\end{figure*}
\end{document}